\documentclass[11pt,a4paper]{article}

\pdfoutput=1

\usepackage[utf8]{inputenc}
\usepackage[T1]{fontenc}
\usepackage{lmodern}
\usepackage{microtype}

\usepackage[margin=1in]{geometry}

\usepackage{amsmath,amssymb,amsthm,mathtools}

\usepackage{graphicx}
\usepackage{float}

\usepackage{xcolor}
\usepackage[colorlinks=true,linkcolor=blue!70!black,citecolor=green!50!black,urlcolor=blue!60!black]{hyperref}
\usepackage{url}
\usepackage[nameinlink,capitalize]{cleveref}

\usepackage{array}

\usepackage{listings}

\newtheorem{theorem}{Theorem}[section]
\newtheorem{lemma}[theorem]{Lemma}
\newtheorem{corollary}[theorem]{Corollary}

\newtheorem{definition}[theorem]{Definition}
\newtheorem{example}[theorem]{Example}
\newtheorem{remark}[theorem]{Remark}

\newcommand{\D}{\texttt{\textbf{d}}}
\newcommand{\eps}{\varepsilon}
\newcommand{\R}{\mathbb{R}}
\renewcommand{\SS}{\mathbb{S}^1}
\newcommand{\Sh}{\mathcal{S}}
\newcommand{\F}{\mathcal{F}}
\newcommand{\Rm}{\mathcal{R}_m}
\newcommand{\dir}{\mathsf{d}}
\newcommand{\RoC}{\mathsf{RoC}}
\newcommand{\DFT}{\mathsf{DFT}}

\newcommand{\Rot}{\mathsf{Rot}}
\newcommand{\ms}{\ell^2(\mathbb{Z}_m)}
\newcommand{\cms}{\mathbb{Z}_m}

\graphicspath{{figures/}}

\title{Rotation-Invariant Vectorized Shape Representations}

\author{%
  Hamid Shafieasl\thanks{University of Utah, USA. Email: \texttt{h.shafieasl@utah.edu}. \url{https://hamidmath.github.io}} \and
  Jeff M. Phillips\thanks{University of Utah, USA. Email: \texttt{jeffp@cs.utah.edu}. \url{http://www.cs.utah.edu/~jeffp}}%
}
\date{}

\begin{document}

\maketitle

\begin{abstract}
We introduce a rotation-invariant representation of planar shapes.  In particular, this representation encodes shapes as vectors such that the Euclidean distance between them serves as a valid shape distance.  For standardized, star-shaped objects, we can deterministically create a sketched vector of dimension $O(1/\eps)$ in $O((1/\eps) \log (1/\eps))$ time that approximates this shape distance to within $\eps$.  
Moreover, because the representation is a standard Euclidean vector, we can directly and efficiently perform various data analyses, such as nearest neighbor search and clustering, in shape space, inherently invariant to the rotation of the shapes.  We demonstrate this through a series of simple experiments.  

The key technical contribution operates on functions over $\mathbb{S}^1$, which we use to encode standardized objects.  The most general rotation-invariant representation of these functions works through a map to an infinite-dimensional function space, parameterized by an offset parameter.  By analyzing special discretized cases of these functions, we show that the representation is strictly injective up to the desired rotation and a mirror-flip-type operation we call \emph{reverse of complement} (RoC).  While RoC status can be controlled by how the function is defined, it is inherent to the representation and required to be handled in the analysis.  Regardless, the vectorized representation is robust to small shape perturbations, and hence discretizing the angles leads to the efficient approximation and algorithm.    
\end{abstract}

\section{Introduction}
\label{sec:intro}

The \emph{shape} of an object (a compact set $X \subset \R^d$) refers to its properties modulo position and orientation.  A core challenge when working in shape spaces~\cite{dryden98,kendall2009shape} is defining a proper distance between shapes, notably, this distance must not depend on the specific positional embedding of the shape. That is, the comparison between two shapes should return the same value under any allowable transformation of one of the shapes.  

There have been three major approaches to this challenge.  
First, and perhaps most popular in computational geometry~\cite{huttenlocher1992dynamic,veltkamp2001state}, is to compute a shape distance $\D(X_1, X_2) = \min_{t \in \R^d, R \in \mathbb{SO}(d)} \D'(X_1, R(X_2 + t))$ where $t$ is a translation~\cite{ahsw-hdtpd-03,ben2014minimum,chan2025linear}, $R$ a rotation~\cite{agarwal2006bipartite,cheng2013shape,chan2024convex}, and $\D'$ is a traditional distance between sets (e.g., Hausdorff~\cite{chew1997geometric} or largest polygon placement~\cite{agarwal1998largest}) which relies on the positional embedding of sets $X_1$ and $X_2$.  For some $L_2$ variants where an alignment is known, the Procrustes approach~\cite{horn1987closed,arun1987least,hanson1981analysis,faugeras1983} offers closed-form solutions. However, the minimization over choices of $t$ and $R$ is typically computationally challenging, although it has been carefully studied and formalized as in the citations above.  
The second approach is to compute intrinsic shape characteristics that do not depend on positional information, and define a distance between these characteristics. Such approaches include, for instance, the turning function between curves~\cite{arkin1991efficiently}, the Gromov-Hausdorff distance~\cite{memoli2008gromov,memoli2012some,schmiedl2017computational}, and the Persistent Homology Transform~\cite{arya2025sheaf,curry2022many}.  These distances again typically require optimizing an alignment between the characteristic-bearing parts of the shapes, making them computationally demanding.  
The third category comprises machine learning-driven heuristics~\cite{xiao2023unsupervised} that compute rotation-invariant representations~\cite{osada2002shape}. Classic examples include those based on spherical harmonics~\cite{kazhdan2003rotation}, while modern methods often use the PointNet architecture~\cite{qi2017pointnet,qi2017pointnet++}. While these representations are empirically robust and translation/rotation invariant, they lack formal analysis regarding which metric properties they preserve and how fine a resolution is needed to distinguish different objects.  A major advantage of these heuristic approaches is that they encode the shape as a vector~\cite{phillips2020sketched}, allowing for fast downstream computation without the need to optimize an alignment.

\paragraph*{This paper} introduces a new representation for planar shapes that is: 
\begin{enumerate}
\item rotation invariant; 
\item a vector, and the shape distance is Euclidean distance between the sketched vectors; 
\item up to a controlled level of discretization, the Euclidean vector distance is a metric between the base shapes, modulo rotation and an inherent $\RoC$ operation; and
\item can $\eps$-approximate standardized star shapes in $O(\frac{1}{\eps})$ dimensions in $O(\frac{1}{\eps} \log \frac{1}{\eps})$ time.  
\end{enumerate}

\paragraph*{Power of vectors.}
We reiterate that the second property means we can embed shapes as (high-dimensional) vectors, and then no longer need to worry about optimizing an alignment between them.  The ability to directly use the Euclidean distance means that this can immediately plug into all of the amazing infrastructure built for this setting.  This includes recent advances in fast nearest neighbor search~\cite{aumuller2020ann} using NHSW structures~\cite{malkov2018efficient,douze2025faiss} (or LSH~\cite{andoni2015practical} if one prefers stronger guarantees).  Moreover one can create mean shape representations using the sketch representation, and thus directly apply $k$-means clustering.  Further this representation is the standard input for any machine learning library for classification or regression.  
While a full exploration of this power is beyond the scope of this paper, we do demonstrate the stability and convenience of this representation through some simple retrieval and clustering applications on $2$-dimensional shapes.

Another important aspect of such sketched representations is the computational cost to construct them.  As with any Euclidean vectors, the dimensionality can be reduced to $O((1/\eps^2) \log n)$ dimensions to preserve relative $(1+\eps)$ error among $n$ shapes, and our method allows us to directly compute such a randomized approximate embedding akin to Rahimi-Recht random Fourier features (RFFs)~\cite{RahimiRecht2007} for the kernel methods~\cite{hofmann2008kernel}.  However, more interestingly, we show how to deterministically create representations efficiently with a fast Fourier transform~\cite{CooleyTukey1965}. This leads to an $\eps$-approximate representation in $O((1/\eps)\log(1/\eps))$ time.

\paragraph*{Rotationally discretizing shapes.}
The core analysis of our method operates in $\SS$, and in fact a discrete version of it $\SS_m$.     
We convert a shape into a function on $\SS_m$ by considering \emph{standardized star shapes}, which are objects where every point in the shape can be connected to the origin by a line segment fully contained within its interior.  We discretize shapes by their polar coordinates uniformly in $m$ wedges, and each distance maps to one value in $\SS_m$; this mapping is common in shape analysis~\cite{marr1977,belongie2002,Mumford1993,Srivastava2016}.

\paragraph*{Sketching discrete $\SS$.}
The core technical result is a rotation-invariant sketch for functions $f$ on $\SS_m$ mapped to $V_f \in \R^m$.  In this setting we can index $f : [m] \to \R$ (more generally $f : [m] \to \R^w$) by an integer $j \in [m] = \{0,1, 2, \ldots, m-1\}$.  
The $k$th coordinate of $V_f$, denoted $V_f(k) \in \R$, encapsulates all pairs of entries in $f$ with a difference of $k$ in their index.  We set $V_f(k) = \sum_{j=1}^m \Phi(f(j) - f(j+k))$ where typically $\Phi(z) = \exp(-z)$.  The characteristic nature of $\Phi$ ensures that the finite sum of the terms in $V_f(k)$ must originate from a unique set $\Delta_f(k)$ of distances.  Moreover, since this only compares offsets, it is invariant to rotation.  The challenging aspect is showing that in whole (using each $V_f(k)$ for $k \in [m]$) this sketch is \emph{only} invariant up to rotation, and no other transformations;  so difficult that it is actually not strictly true!  It also permits a constant offset operation, as well as a natural new sort of mirror-flip transformation we call \emph{reverse-of-complement (RoC)}.  However, these operations are fairly simple, and the offset can be controlled by setting the scale in $V_f$ (based on how we derive them from shapes).

\section{Preliminaries: Families of Shapes and Functions Considered}
\label{sec:prelim}

To make our theorems precise, we describe some specific families of shapes and functions.  

\paragraph*{Shape families.}
We consider \emph{standardized shapes} $\Sh$ which are compact sets $X$ in $\R^d$ which have been mean centered so $\bar x = \int_{x \in X} x$ is the origin $\textbf{0} = (0,\ldots, 0)$, and has bounded radius of the shape: $\max_{x \in X} \|x\| \leq 1$.  Any shape can be converted into this form by a simple standardization process: (1) compute its mean point $\bar x$ and subtract it from all coordinates, and (2) compute its maximum radius $\tau$ and divide all coordinates by $\tau$.  
As a further restriction, we sometimes also consider \emph{standardized star shapes} $\Sh_\star$.  These are standardized shapes $X \in \Sh$ that are also star-shaped with respect to the origin. This means all points $x \in X$ can be connected to the origin $\textbf{0}$ with a line segment that does not leave the shape.  

Note that by considering standardized shapes, the choice of translation (and scaling) is removed.  This allows us to focus on the rotational-invariance property.  

We focus on $d=2$ dimensions, and consider a discrete version of rotation; for some positive integer parameter $m$, let $\mathcal{R}_m = \mathbb{SO}(2)_m$ be rotations at $j \frac{2 \pi}{m}$ radians for $j \in [m] = \{0, 1, 2, \ldots, m-1\}$.  We pair this rotational restriction with shape families $\Sh^m$ which are assumed to be uniform within $m$ wedges around the origin.  Specifically, any ray $r_\theta$ from the origin, with angle of $\theta \in \Theta_j = \left[j \frac{2 \pi}{m}, (j+1) \frac{2\pi}{m}\right)$ from the horizontal direction, has the same intersection with $X$.  That is, let $I(r_\theta \cap X)$ describe a subset of $[0,1]$ (recall standardized shapes ensure the extent of the intersection is at most a distance $1$; and for shape in $\Sh_\star$ this subset is an interval $[0, a_\theta]$ for some $a_\theta \in (0,1]$), and for two $\theta, \theta' \in \Theta_j$ then $I(r_\theta \cap X) = I(r_{\theta'} \cap X)$.  This restriction ensures that if two shapes $X_1, X_2 \in \Sh^m$ are identical under some rotation $R \in \mathbb{SO}(2)$ (i.e., $X_1 = R(X_2)$) then, this rotation must be in the discrete rotation set $\Rm$.  And for any shape $X \in \Sh^m$ and rotation $R \in \Rm$, then we also have $R(X) \in \Sh^m$ (it forms a group).  

Moreover, let us consider \emph{$\eta$-Lipschitz} star-shaped standardized shapes $\Sh_{\star,\eta}$; these $X \in \Sh_{\star,\eta}$ ensure that for two rays $r_\theta, r_{\theta'}$ from the origin that define intersections $I(r_\theta \cap X) = [0,a_\theta]$ and $I(r_{\theta'} \cap X) = [0,a_{\theta'}]$ satisfy that $|a _\theta- a_{\theta'}| \leq \eta \frac{1}{2\pi}|\theta - \theta'|$.  
For such  \emph{$\eta$-Lipschitz} shapes $X \in \Sh_{\star,\eta}$, then we will be able to show that for any error parameter $\eps$ and a distance $\D$, when $m = O(1/\eps)$ we can deterministically construct an $X' \in \Sh_\star^m$ such that $\D(X, X') \leq \eps$.

\paragraph*{Function families on $\SS$.}
As previewed, we convert these shapes in a rotational invariant way into functions on $\SS$.  
We restrict to a discrete version of $\SS$, denoted $\SS_m$ which is indexed by $m$ values in $[m] = \{0, 1, 2, \ldots, m-1\}$.  
We again restrict the family $\F$ of $2\pi$-periodic functions, so $f \in \F$ is  $f : \SS \to \mathbb{A}$, where $\mathbb{A}$ is the set of real algebraic numbers.  
We also consider two easy to analyze families.  
Let $\F^m_{\neq}$ be the family where all values $f(j)$ are distinct, and moreover all pairwise differences $f(j) - f(j')$ for $j\neq j' \in [m]$ are also distinct.  
Second, let $\F^m_\sigma$ be the family where $f : [m] \to [m]$ and $f(j) \neq f(j')$ for $j \neq j'$; this describes all permutations of the values in $[m]$.  These families lead to different types of analysis: $\F^m_{\neq}$ leverages uniqueness, while $\F^m_{\sigma}$ contains many tied difference values but possesses structural symmetry that aids case analysis.

\section{Sketching $\SS$}
\label{sec:sketch-S1}

In this section, we describe how to sketch functions $f \in \F$ on $\SS$ to a function space $\psi_f \in \Psi$, which can be (approximately) compactly represented as vectors $V_f \in \R^t$.  In particular, these representations will be invariant to rotations in $\SS$.  

\paragraph*{The Sketch.}
We will first define in the continuous setting, with infinite dimensional functional representation, but then focus on the easier to work with discrete case:
For parameter $\alpha \in [0, 2\pi)$ define 
\[
 \psi_f(\alpha) = \frac{1}{2\pi} \int_{\theta \in [0,2\pi)} \Phi(f(\theta) - f(\theta+\alpha)) \dir \theta.
\]
The mapping function $\Phi$ we consider will primarily be $\Phi(z) = \exp(-z)$.   However it will also be interesting to consider kernel functions such as the Laplace $\Phi(z) = \exp(-|z|)$ and Gaussian $\Phi(z) = \exp(-z^2)$. 
This can be interpreted as considering all function $f$ values offset by $\alpha$, processing them through a characteristic function $\Phi$, and averaging over all $\theta$.  

When considering $f \in \F^m$ and only rotations in $\Rm$, then there are only $m$ choices for $\theta$ and for $\alpha$ so the mapping simplifies to a result in a vector $V_f \in \R^m$ defined so the $j$th coordinate (representing angle $\theta = j \frac{2 \pi}{m}$), for $j \in [m]$ is 
\[
V_f(k) = \frac{1}{m} \sum_{j \in [m]} \Phi(f(j) - f(j + k))
\]
Recall that since $f \in \F^m$ is described on indexes in $[m]$, but the input $j+k$ of the second term might can be in $[2m]$, that $f$ takes inputs modulo $m$, so $f(m+i) = f(i)$.

\subsection{Invariances of the Sketch}

We first provide a straightforward observation regarding invariance to rotations and shifts.  

\begin{lemma}[Rotation and Shift Invariance]\label{rot_shift_invariance}
    For $f \in \F$, any rotation $\beta \in [0,2\pi)$ and any constant $c \in \R$, define $g \in \F$ as
    $
    g(\alpha) = f(\alpha + \beta) + c.  
    $
    Then the signature map satisfies
    $\psi_g(\alpha) = \psi_f(\alpha)$  for all  $\alpha \in [0,2\pi)$.
\end{lemma}
\begin{proof}
By definition, for all $\alpha\in [0,2\pi)$ we have
\begin{align*}
\psi_g(\alpha) 
  = 
  \frac{1}{2\pi} \int_0^{2\pi} \Phi\Bigl(g(\theta) - g(\theta+\alpha)\Bigr)d\theta 
& = 
\frac{1}{2\pi}\int_0^{2\pi} \Phi\Bigl(f(\theta + \beta) + c - \bigl(f(\theta + \alpha + \beta) + c\bigr)\Bigr)d\theta\\
& = 
 \frac{1}{2\pi}\int_0^{2\pi} \Phi\Bigl(f(\theta + \beta) - f(\theta + \alpha + \beta)\Bigr)d\theta\\
\{\text{Substitute } \gamma = \theta+\beta\}\qquad 
& = 
\frac{1}{2\pi}\int_{\beta}^{2\pi+\beta} \Phi\Bigl(f(\gamma) - f(\gamma + \alpha)\Bigr)d\gamma\\
\{f \text{ is a } 2\pi\text{-periodic function}\}\  \qquad 
& = 
\frac{1}{2\pi}\int_{0}^{2\pi} \Phi\Bigl(f(\gamma) - f(\gamma + \alpha)\Bigr)d\gamma= \psi_f(\alpha). \qedhere
\end{align*}
\end{proof}

\paragraph*{Reverse of Complement.} 
We next define a less intuitive function operation, which we call \emph{reverse of complement (RoC)}.  Let $\F_M$ be the set of functions with $f(\theta) \in [0,M]$.  Then $f' = \RoC(f)$ satisfies $f'(\theta) = M - f(2\pi - \theta)$.  In fact, this operation can use any value for $M$ (including $M=0$), but will be useful later to define this way when considering $f \in \F_M$.

\begin{lemma}[RoC, rotation and shift invariance]
\label{lem:involution_shift_invariance}
    For $f \in \F$, any rotation $\beta' \in [0,2\pi)$ and any constant $c' \in \R$, define $g \in \F$ as
    $
    g(\alpha) = \RoC(f(\alpha + \beta')) + c'.  
    $
    Then the signature map satisfies
    $\psi_g(\alpha) = \psi_f(\alpha)$  for all  $\alpha \in [0,2\pi]$.
\end{lemma}

\begin{proof}
First observe that we can simplify 
\[
g(\alpha) = \RoC(f(\alpha + \beta')) + c' = M - f(2 \pi - \alpha - \beta') + c'.  
\]
Now setting $\beta = 2 \pi - \beta'$ and $c = c' + M$ we have $g(\alpha) = -f(-\theta + \beta) + c$.  Hence
\[
\begin{aligned}
\psi_g(\alpha)
&= \frac{1}{2\pi}\int_0^{2\pi}\Phi\bigl(g(\phi)-g(\phi+\alpha)\bigr)\,d\phi \\
&= \frac{1}{2\pi}\int_0^{2\pi}\Phi\bigl(-f(-\phi+\beta)+c -\bigl(-f(-(\phi+\alpha)+\beta)+c\bigr)\bigr)\,d\phi\\
&= \frac{1}{2\pi}\int_0^{2\pi}\Phi\bigl(-f(-\phi+\beta) + f(-\phi-\alpha+\beta)\bigr)\,d\phi\\
&= \frac{1}{2\pi}\int_0^{2\pi}\Phi\bigl(f(-\phi-\alpha+\beta) - f(-\phi+\beta)\bigr)\,d\phi.
\end{aligned}
\]

Make the change of variables \(\zeta := -\phi+\beta\). Then \(d\zeta = -\,d\phi\). As \(\phi\) runs over \([0,2\pi)\), $\zeta$ runs over an interval of length \(2\pi\) (specifically \([\beta,\beta-2\pi)\)), and by \(2\pi\)-periodicity we may replace that interval by \([0,2\pi)\). Applying the change of variables,
\[
\begin{aligned}
\psi_g(\alpha)
&= \frac{1}{2\pi}\int_{\beta}^{\beta-2\pi}\Phi\bigl(f(\zeta-\alpha)-f(\zeta)\bigr)(-d\zeta)
= \frac{1}{2\pi}\int_{\beta-2\pi}^{\beta}\Phi\bigl(f(\zeta-\alpha)-f(\zeta)\bigr)\,d\zeta\\
&= \frac{1}{2\pi}\int_0^{2\pi}\Phi\bigl(f(\zeta-\alpha)-f(\zeta)\bigr)\,d\zeta.
\end{aligned}
\]
Finally, relabel the dummy integration variable (or shift \(\zeta\mapsto\zeta+\alpha\)) to obtain
\[
\psi_g(\alpha)=\frac{1}{2\pi}\int_0^{2\pi}\Phi\bigl(f(\theta)-f(\theta+\alpha)\bigr)\,d\theta = \psi_f(\alpha).
\]
This proves the claim.
\end{proof}

Note that, when we also consider shift and rotational invariance, it is equivalent to think of $\RoC$ as a form of \emph{involution} $\mathsf{Inv}(g(\alpha)) = -f(-\alpha)$.  
Moreover, the above proofs work for $\F^m$ over $\Rm$; for the other direction which we describe next we restrict to this setting.  

\paragraph*{Periodic convolution.}   
A periodic convolution of two functions $p,q$  on $\SS$ is defined $(p * q)(\alpha) = \int_0^{2\pi} p(\theta)q(\alpha-\theta) \dir \theta$.  We show  $\psi(\alpha)$ can be decomposed this way.  

\begin{lemma}\label{lem:convolution}
With $\Phi(z) = \exp(-z)$, then $\psi(\alpha)$ is a periodic convolution of $p_f(\alpha) = \exp(-f(-\alpha))$ and $q_f(\alpha) = \exp(f(\alpha))$.  
\end{lemma}
\begin{proof}
\begin{align*} \label{eq:convolution}
(p_f * q_f)(\alpha) 
 &= 
\int_{\theta = 0}^{2\pi} p_f(\theta) q_f(\alpha -\theta) \dir \theta 
= 
\int_{\theta = 0}^{2\pi} \exp(-f(-\theta)) \exp(f(\alpha - \theta)) \dir \theta
\\ & = 
\int_{\theta = 0}^{2\pi} \exp(-(f(\theta) - f(\alpha-\theta))) \dir \theta 
\\ & = 
\int_{\theta' = 0}^{2\pi} \exp(-(f(\theta') - f(\alpha + \theta'))) \dir \theta' 
= 
2\pi \cdot \psi_f(\alpha). 
\end{align*}
The second to last equality can be seen by substituting $\theta' = -\theta$, which does not change the periodic $[0,2 \pi)$ domain.  
\end{proof}

This allows us to analyze $\psi_f$ using the powerful tools of convolution theory, such as the Convolution Theorem \cite[Sections~3.4~\&~8.2]{Oppenheim2009}.  
And in particular, implies that for any $f \in \F$ if we define $g(\alpha) = -f(-\alpha)$ then $\psi_f = \psi_g$; this implies the invariance to $\RoC$, by interpreting it as an involution $\mathsf{Inv}$.

\subsection{Injectivity of the Sketch}
We next show that for a discrete class of functions $\F^m$ that the sketch is injective in that it is \emph{only} invariant up to the rotation, constant offsets, and $\RoC$ operations.  

For a function $f \in \F^m$, let $\Delta_f(k) = \{ f(j) - f(j+k) \mid j \in [m]\}$ be the \emph{multiset} of function differences offset by $k \in [m]$.  Recall that we can define $V_f(k) = \frac{1}{m}\sum_{z \in \Delta_f(k)} \Phi(z)$.  We show that these $\{\Delta_f(k)\}_{k \in [m]}$ are an isomorphic representation of the sketch $V_f$.  

\begin{lemma}\label{lem:exp-iso}
    Consider two multisets $W, W'$ of algebraic real numbers with $|W|=|W'|=m$, a rational number $\lambda \ne 0$ and fix $p\in \mathbb{Z}$.   If $\sum_{w_i \in W} \exp(\lambda w_i^{2p+1}) = \sum_{w'_i \in W'} \exp(\lambda {w'}_i^{2p+1})$, then $W = W'$. 
\end{lemma}
\begin{proof}
    Let $\mathcal{A}$ be the set of algebraic numbers. For $\alpha\in\mathcal{A}$, let $W(\alpha)=\{w\in W\ | w = \alpha\}$ (resp. $W'(\alpha)$).  Then define $\delta_\alpha = |W(\alpha)| - |W'(\alpha)|$.  Let $L=\{\delta_\alpha\ |\ \delta_\alpha\ne 0, \alpha\in\mathcal{A}\}$. The equality condition in our lemma ensures that $\sum_{\alpha\in\mathcal{A}} |W(\alpha)|\exp(\lambda\alpha^{2p+1}) = \sum_{\alpha\in\mathcal{A}} |W'(\alpha)|\exp(\lambda\alpha^{2p+1})$. Therefore, 
    $\sum_{\alpha\in\mathcal{A}} \delta_\alpha \exp(\lambda\alpha^{2p+1}) = 0$.  Our goal is to show that $L = \emptyset$.
    
    We now invoke a classical result from transcendental number theory, the
    \emph{Lindemann--Weierstrass theorem} (\cite[Theorem~1.4]{baker1975}),
    which states that if $\alpha_1, \ldots, \alpha_n$ are distinct algebraic numbers, 
    then $e^{\alpha_1}, \ldots, e^{\alpha_n}$ are linearly independent over~$\mathbb{Q}$.
    
    Therefore, the finite set 
    \(
    \{\, \exp(\lambda \alpha^{2p+1}) : \delta_\alpha \neq 0 \,\}
    \)
    is linearly independent, which forces $\delta_\alpha = 0$ for all such $\alpha$. 
    Hence $L = \emptyset$.
\end{proof}

\begin{corollary}  \label{cor:Delta(k)}
    For two functions $f,g \in \F^m$ with algebraic real values, then when using $\Phi(z) = \exp(-z)$, then $V_f = V_g$ is equivalent to have $\Delta_f(k) = \Delta_g(k)$ for all $k \in [m]$.
\end{corollary}

We say two functions $f,g \in \F^m$ so that $\Delta_f(k) = \Delta_g(k)$ for all $k \in [m]$ are \emph{lag-homometric}.  Since $\F^m, \Rm$ is a special case of $\F,\mathbb{SO}(2)$ (as elaborated on in Appendix \ref{app:lag-homomorph}), then Lemma \ref{lem:involution_shift_invariance} established that lag-homometric functions are in the same equivalence class up to rotations, offsets, and potential $\RoC$ operations.

\paragraph*{Canonical representations.}
We first establish a representation of a function $f \in \F^m$ fixing the choice of rotation and constant offsets.  
We say a $f \in \F$ is \emph{canonical} if its smallest value is $0$, and this value is in the first position so $f(0) =0$.  We can enforce the first condition by identifying $\alpha = \min_j f(j)$, and then, through a constant offset, subtracting $\alpha$ from each entry in $f$ so $f(j) \leftarrow f(j) - \alpha$.  If $j_{\min} = \arg\min_j f(j) \neq 0$, then we need to perform a rotation.   We set $f(j) \leftarrow f(j -_m j_{\min})$ where $-_m$ indicates subtraction modulo $m$.  

\begin{lemma}\label{lem:canonical}
    Let $f,g \in \F^m_{\neq}$ be lag-homometric (which includes e.g., $g = \RoC(f)$).  If $\tilde f, \tilde g$ are the canonical forms of $f,g$, then
    \[
    \max_{j\in [m]} \tilde f(j) = \max_{j \in [m]} \tilde g(j) 
    \;\;\;\; \text{ and } \;\;\;\;
    \arg\max_{j\in [m]} \tilde f(j) = \arg\max_{j \in [m]} \tilde g(j).
    \]
\end{lemma}
\begin{proof}
    Let the maximum value $\beta = \max_{j\in [m]} \tilde f(j)$, and $j_{\max} = \arg\max_{j\in [m]} \tilde f(j)$.  Since $\tilde f(0) = 0$ then $\beta = \tilde f(j_{\max}) - \tilde f(0) \in \Delta_f(m-j_{\max})$.  Since $f,g$ are lag-homomorphic, which by Lemma \ref{lem:involution_shift_invariance} and Lemma \ref{lem:exp-iso} implies so are $\tilde f, \tilde g$, since they only change by rotation and shift.  Now because $\tilde g(0) = 0$, and since $\beta \in \Delta_g(m-j_{\max})$, there must exist some $j$ with $\tilde g(j) - \tilde g(j+m-j_{\max}) = \beta$.  Since $\beta$ is the global maximum difference across all shifts, and all differences of $\tilde g$ are at most $\max \tilde g - \min \tilde g = \max \tilde g$ (because $\min \tilde g = 0$), we must have $\max \tilde g \geq \beta$.  But $\beta$ must also be the maximum difference in $g$ (since all multisets are shared), so $\max \tilde g = \beta$ and $j_{\max} = \arg\max_{j \in [m]} \tilde g(j)$.
\end{proof}

This argument on $\F^m_{\neq}$ only used the uniqueness of the maximum.  Thus it also applies to the other function family we consider on permutations $\F^m_{\sigma}$.  From here, the arguments needed to prove that $V_f = V_g$ implies that $f$ matches $g$ up to rotation, shifts, and $\RoC$ differs for this permutation family.  These arguments are technical and leverage the strong symmetry of the sets of differences.  We defer showing this invariance of $\F^m_\sigma$ to Appendix \ref{app:permutations-inv}.

Next we resolve the choice of $\RoC$ operation.  Let $\Delta_f = \bigcup_{k \in [m]} \Delta_f(k)$.
Let $\beta_2$ be the second largest values in $\Delta_f$ (recall $\beta$ is the largest).  
We say $f$ is in \emph{$\RoC$-canonical form} (so $f(0) =0$) then let $j_1 = \arg\max_{j \in [m]} f(j)$ and $j_2 = \arg\max_{j \in [m], j \neq j_1} f(j)$, then $f(j_1) = \beta$ (by Lemma \ref{lem:canonical}) and $f(j_2) = \beta_2$.  

\begin{lemma}\label{lem:RoC-canon}
    If $f\in \F^m_{\neq}$ we can always transform it into $\bar f$ in $\RoC$-canonical form using shift, rotation, and $\RoC$ operations.
\end{lemma}
\begin{proof}
    First transform $f$ into canonical form $\tilde f$, and define $\beta$, $\beta_2$ with respect to $\tilde f$.  Now the second largest difference $\beta_2$ can either be between $f(0) = 0$ and the second largest value of $f$ (at $j_2$), or between the largest value $f(j_1)$ and the second smallest value of $f$ at $j'_2$.  If it is the first case then $\beta_2 = f(j_2) - f(0) = f(j_2)$ and we are in $\RoC$-canonical form.  If it is the second case, then we apply an $f' = \RoC_\beta(\tilde f)$ operation, and then make $f'$ canonical as $\tilde f'$.  This $\RoC_\beta$ operation makes $f'((m-1)-j_1) = 0$, $f'(m-1) = \beta$, and $f'((m-1)-j'_2)$ becomes the second largest value.  Importantly, it preserves the second largest difference $\beta_2 = f'((m-1)-j'_2) - f((m-1)-j_1) = f'((m-1)-j'_2)$, satisfying $\RoC$ canonical form (after rotation to canonical form) since $f'((m-1)-j_1) = 0$.
\end{proof}

\paragraph*{Inductive argument.}
We are now ready for the inductive argument to show that once lag-homometric functions $f,g \in \F^m_{\neq}$ are put in $\RoC$-canonical form they must be equal.  
For this argument we introduce some more notation.  Sort the values in $f$ and in $g$ and represent their sorted order as $\hat f, \hat g$; this format implies $\hat f(j) < \hat f(j+1)$ for all $j \in [m-1]$, $\hat f(j) = f(j')$ for some $j,j' \in [m]$, and $\hat f(0) = 0$; and the same for $\hat g$.  

Now we define a subset of differences $\Delta_f^{s,t} \subset \Delta_f$.  It excludes all differences formed by the $j \leq s$ smallest values and the $j' \leq t$ largest values of $f$ (it can pair two large or two small values).  When we have established the location of the $s$ smallest and $t$ largest values in both $f$ and $g$, then we have resolved these differences, and can examine the largest difference remaining in $\Delta_f^{s,t} = \Delta_g^{s,t}$.  Note that by Lemma \ref{lem:canonical} and the definition of $\RoC$-canonical form, when we have this form for $f,g$ then we have resolved the $s=1$ smallest value and $t=2$ largest distance.  

We now state the most technical lemma, proved in Appendix \ref{app:unique-inv} as Lemma \ref{lem:Dst-induction-app}.  

\begin{lemma} \label{lem:Dst-induction}
    Consider $f,g \in \F^m_{\neq}$ which are lag-homometric and in $\RoC$-canonical form.  Assume we have resolved the locations of the $s \geq 1$ smallest and $t \geq 2$ largest values in $f,g$, confirming they are the same.  We can then confirm the location of either the $(s+1)$th smallest or $(t+1)$th largest by value of $f,g$ and confirm that this value matches.  
\end{lemma}
\begin{proof}[Proof Sketch]
  Each of $f$ and $g$ has two options for which difference is the largest in $\Delta_f^{s,t} = \Delta_g^{s,t}$; using the $(s+1)$th smallest or $(t+1)$th largest value.  If these match, we satisfy the lemma's claim.  Showing if they do not match it causes a contradiction is technical.  We pick a location where $f$ and $g$ differ in this case, and compare the difference with $\hat f(m-2) = \hat g(m-2)$, which generates different differences $\delta_f, \delta_g$ both placed in some $\Delta_f(k)$ and $\Delta_g(k)$, respectively.  Since these sets $\Delta_f(k)$ and $\Delta_g(k)$ must match, and all difference are unique, we can argue that we cannot form $\delta_g$ to put in $\Delta_f(k)$, and they cannot match.      
\end{proof}

Starting with the base case of $s=1, t=2$ from $\RoC$-canonical form, we can invoke Lemma \ref{lem:Dst-induction} inductively.  Each time it increases either $s$ or $t$.  This continues until $s+t = m$, in which case all positions of $f$ and $g$ are confirmed to be the same, completing the proof.   This results in the following theorem.  

\begin{theorem}\label{thm:fg-same}
  For $f,g \in F^m_{\neq}$ which are lag-homometric in $\RoC$-canonical form, then they must be equal.      
\end{theorem}

\begin{theorem} \label{thm:main-invariance}
    Consider two functions $f,g \in \F^m_{\neq}$. Then $V_f = V_g$ if and only if there exist a rotation $R \in \Rm$ and a constant $c \in \R$ such that $f = R(g)+c$ or $f = R(\RoC(g))+c$.  
\end{theorem}
\begin{proof}
   The forward implication follows from Lemma \ref{lem:involution_shift_invariance}.  
   
   For the converse, Corollary \ref{cor:Delta(k)} shows $V_f=V_g$ implies that $f,g$ are lag-homometric.  
   Then Lemma \ref{lem:RoC-canon} shows we can transform both $\bar f \leftarrow f$ and $\bar g \leftarrow g$ so both $\bar f$ and $\bar g$ are in $\RoC$-canonical forms and still lag-homometric with each other and $f,g$.  Then by Theorem \ref{thm:fg-same} we must have $\bar f = \bar g$.  Hence, any other operation to $g$ would change $\bar g$ violating it being lag-homometric, and would hence change the representations $V_g$, completing the proof.  
\end{proof}

Let $[f]$ denote the equivalence class of functions to $f \in \F^m_{\neq}$ up to shifts, rotations, and $\RoC$.  Let $\F^m_{[\neq]}$ be the family of distinct functions modulo this equivalence class.  
Define function $\D'_\Phi(f,g) = \|V_f - V_g\|$ for $V_f, V_g$ being the sketched representation of $f,g$ respectively.  

\begin{corollary}
    The distance $\D'_\Phi$ for $\Phi(z) = \exp(-z)$ is a \textbf{metric} on $\F^m_{[\neq]}$.
\end{corollary}

\subsection{Inner Products for the Sketch}
\label{sec:kernel}

We now define an inner product on the space of the sketch signatures, which is the standard inner product of a Hilbert space and quantifies the similarity between two functions in $\F$.

The \emph{invariant sketch kernel} \(K(f,g)\) is defined as the inner product between two signature functions $\psi_f$ and $\psi_g$ for $f,g \in \F$ as  
\[
K(f,g) = \langle \psi_f, \psi_g \rangle_{L^2} = \frac{1}{(2\pi)^2}\int_0^{2\pi} \psi_f(\alpha) \psi_g(\alpha) \, d\alpha.
\]

The invariant sketch kernel $K(f,g) = \langle \psi_f, \psi_g \rangle_{L^2}$ provides a powerful tool for quantifying the similarity between functions in $\F$.  By measuring the inner product of signature functions $\psi_f$ and $\psi_g$, the kernel captures the degree of alignment between $f$ and $g$ across all possible shifts on the unit circle \(\SS\). 
A common goal for a kernel to be used effectively within kernel methods in machine learning~\cite{schmiedl2017computational}, is for it to be \emph{positive semi-definite}.

\begin{lemma}[Positive Semi-Definiteness of the Sketch Kernel]
For functions $f, g \in \F$, the invariant sketch kernel $K(f,g)$ is positive semi-definite.  
\end{lemma}

\begin{proof}
A positive semidefinite kernel is one that for any finite set of functions $\{f_1, f_2, \dots, f_n\} \subset \F$ and real coefficients $\{c_1, c_2, \dots, c_n\} \subset \R$, then 
$\sum_{i=1}^n \sum_{j=1}^n c_i c_j K(f_i, f_j) \ge 0$.

We show this by first substituting our kernel into the sum
$
\sum_{i=1}^n \sum_{j=1}^n c_i c_j \langle \psi_{f_i}, \psi_{f_j} \rangle_{L^2}.  
$
The inner product \(\langle \cdot, \cdot \rangle_{L^2}\) is bilinear. This property allows us to move the summations and coefficients inside the inner product
$
= \left\langle \sum_{i=1}^n c_i \psi_{f_i}, \sum_{j=1}^n c_j \psi_{f_j} \right\rangle_{L^2}
$.
Now for a new function \(V \in L^2([0, 2\pi])\) that is the linear combination of the signature functions
$
V(\alpha) = \sum_{i=1}^n c_i \psi_{f_i}(\alpha), 
$
this becomes the inner product of \(V\) with itself as
$\langle V, V \rangle_{L^2}= \|V\|^2_{L^2}$.  
Therefore, as desired
\[
\sum_{i=1}^n \sum_{j=1}^n c_i c_j K(f_i, f_j) = \|V\|^2_{L^2} \ge 0. \qedhere
\]
\end{proof}

%

%

\subsection{Computing the Sketch}
For computation, we focus on $f \in \F^m$ over $\Rm$.  
The sketch takes $O(m^2)$ time to compute directly, as it provides a vector in $\R^m$, and each coordinate takes $O(m)$ time to sum over all discrete rotations in $[m]$.  In this section we provide two alternatives.

\paragraph*{Fast FFT algorithm: near-linear time.}
First we leverage that, via Lemma \ref{lem:convolution}, we can write the sketch vector $V_f \in \R^m$, defined for any coordinate $j \in [m]$ as the convolution of two discrete functions $p_f(j) = \exp(-f(-j \frac{2\pi}{m}))$ and $q_f(j) = \exp(f(j \frac{2 \pi}{m}))$.   Specifically:
\[
V_f(j) = \frac{1}{m} (p_f * q_f)(j).
\]
Then we can leverage the \emph{discrete convolution theorem}~\cite{CooleyTukey1965,Oppenheim2009}.  This states that given two discrete signals $p,q \in \R^m$ that their convolution can be computed in the frequency domain
\[
(p * q)[j] = \DFT^{-1}( \DFT(p) \odot \DFT(q) )[j]
\]
where $\DFT$ is the Discrete Fourier Transform, $\DFT^{-1}$ is its inverse, and $\odot$ denotes element-wise multiplication of the transformed signals.
Moreover, this approaches computes all $m$ coordinates simultaneously:  
\begin{enumerate}
    \item Compute the DFT of the signals $p_f$ and $q_f$.
    \item Perform a single element-wise multiplication of their transformed representations.
    \item Compute the Inverse DFT of the resulting product to obtain $\psi_f$.
\end{enumerate}
Using the Fast Fourier Transform (FFT)~\cite{CooleyTukey1965}, can compute $\DFT$ and its inverse $\DFT^{-1}$, each in $O(m \log m)$ time.  

\begin{theorem}  \label{thm:FFT-sketch}
    For $f \in \F^m$ we can build the sketch $V_f \in \R^m$ using $\Phi(z) = \exp(-z)$ in $O(m \log m)$ time.  
\end{theorem}

\paragraph*{Randomized low-dimensional sketch.}
We also note that given two functions $f,g \in \F$, one can view $K(f,g) = \frac{1}{(2\pi)^2} \int_0^{2\pi} \psi_f(\alpha) \psi_g(\alpha) \dir \alpha$ as an expected value of $\alpha \in \mathsf{Unif}(0, 2 \pi]$.  This implies one can estimate $K(f,g)$ from a sample of $t = O(\frac{1}{\eps^2} \log \frac{1}{\delta})$ so with probability at least $1-\delta$, it guarantees at most $\eps M^2$ error for $f,g \in \F_M$.  
In particular, by consistently using $t \ll m$ samples $\alpha_1, \alpha_2, \ldots, \alpha_t \in (0,2\pi]$, this induces an approximate sketch $\hat V_f \in \R^t$.  This is akin to Random Fourier Features~\cite{RahimiRecht2007}; details are in Appendix \ref{app:kernels}.  However the same sort of guarantees can be achieved by applying a JL random projection~\cite{matouvsek2008variants} on $V_f \in \R^m$.  Since our random sketch takes $O(tm)$ time to build, this is not typically much faster than creating the full sketch $V_f \in \R^m$ and applying the Fast JL~\cite{kane2014sparser} in $O(m \log m + t/\eps)$ time.


\section{Mapping from Shapes to Functions in $\SS$}
\label{sec:shape2func}

We consider families of standardized shapes $\Sh^m$ which are discretized into $m$ angular regions.  For $X \in \Sh^m$ this assumes the ray $r_\theta$ from the origin at angle $\theta$ is consistent in its intersection with $X$ within a discrete set of angles $\Theta_j = [j \frac{2 \pi}{m}, (j+1) \frac{2 \pi}{m})$.  Then for star-shaped shapes $X \in \Sh^{\star}_m$ this intersection $I(r_\theta \cap X)$ is further restricted to be of the form of $[0,a]$ for some interval with $a \leq 1$.  

In practice (as we explore in Section \ref{sec:exp}), continuous shapes can be converted into this form by measuring along a discrete set of rays.  

\subsection{Star-Shaped Standardized Shapes}
\label{sec:star-sh}

We first discuss star-shaped standardized shapes from $\Sh^m_\star$.  Such a shape can be isometrically mapped into a vector in $\F^m$.  For $X \in \Sh^m_\star$ for each angular region $\Theta_j$ we determine the intersection $[0,a_j] = I(X \cap r_\theta)$ (for $\theta \in \Theta_j$) with the values $a_j$ of the extent of the shape in that direction.  For function $f \in \F^m$ we assign $f(j) = a_j$ as in \cite{Mumford1993}.  Figure \ref{fig:star-shape} shows an example of a fish shape $X \in \Sh$ approximated with various values of $m$ in $\Sh^m_\star$, and then each represented as its corresponding function in $\F^m$.  

\begin{figure}[H]
    \centering
     \includegraphics[width=\linewidth]{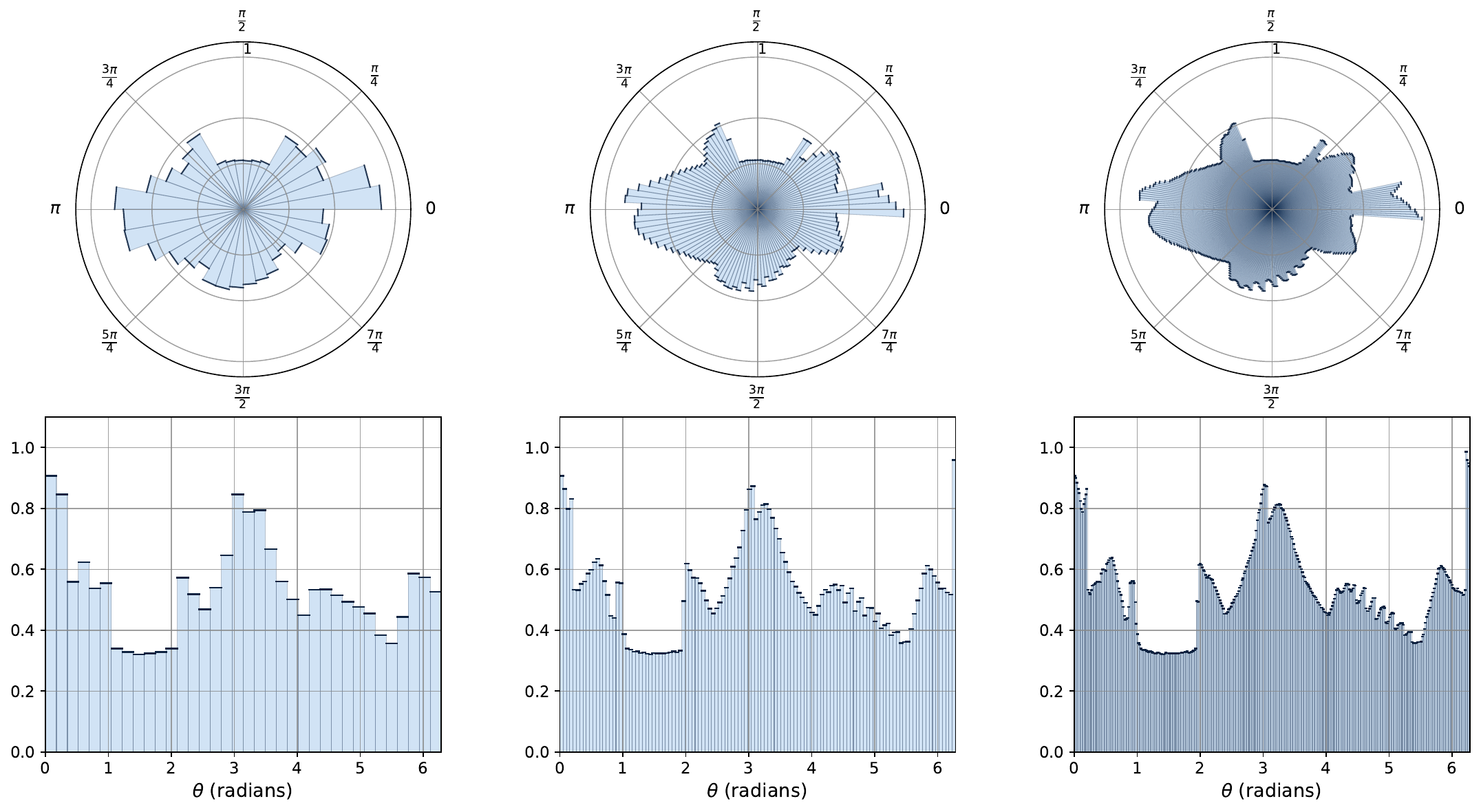}
    \caption{Comparison of a shape at low resolution ($m=36$, left), medium resolution ($m=120$, middle) and high resolution ($m=360$, right). Each panel shows the polar plot (top) of $X \in \Sh_\star^m$, and its corresponding radial signal in $\F^m$ (bottom).}
    \label{fig:star-shape}
\end{figure}

Define our distance between standardized star shapes $X_1, X_2 \in \Sh^m_\star$ via their sketched representations $V_1, V_2 \in \R^m$ of their corresponding function representations $f_1, f_2 \in \F^m$, respectively as 
\[
\D_\Phi(X_1, X_2) = \|V_1 - V_2\|.
\]
We say a shape $X \in \Sh^m_\star$ is in \emph{general position} if both all values $a_j = X(\theta_j)$ are distinct, and all differences $a_j - a_{j'}$ for $j \neq j' \in [m]$ are also distinct.  
Define $[X]$ as the equivalence class of standardized star shapes $X' \in \Sh^m_\star$ that can be rotated $R(X') = X$ or $R(\RoC(X')) = X$ for some $R \in \R_m$.  Note that because there is an invariance between $\Sh^m_\star$ and $\F^m$ (i.e., with a invertible map $\xi : \Sh^m_\star \to \F^m$), then we can inherit the definition of $\RoC$ on $\F^m$ to that on $\Sh^m_\star$; define $\RoC(X) = \xi^{-1}(\RoC(\xi(X)))$.  An illustration of this on a fish is shown in Figure \ref{fig:RoC-fish}.  
Let $\Sh^m_{[\star]}$ as the space of these standardized star shapes modulo rotation and $\RoC$, i.e., each element of equivalence class $[X]$ are considered the same object.

\begin{figure}[H]
    \centering
    \includegraphics[width=\linewidth]{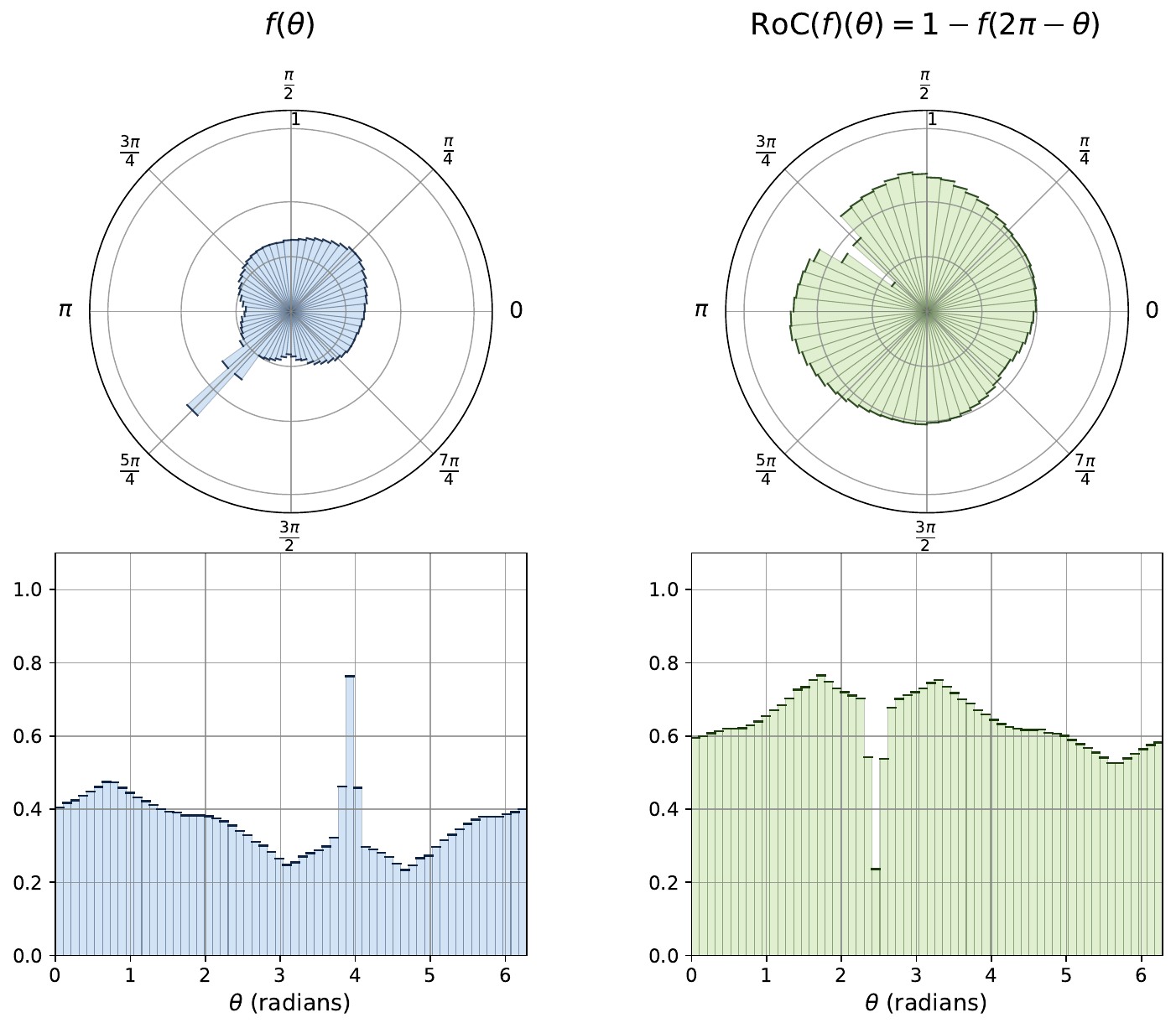}
    \caption{Illustration of an $\RoC$ operation on a fish shape $X$, and its corresponding function $f$.  The original is shown on the left, and the $\RoC(X)$ and $\RoC(f)$ on the right. }
    \label{fig:RoC-fish}
\end{figure}

\begin{theorem}
\label{thm:metric}
For shapes $\Sh^m_{[\star]}$ in general position then $\D_\Phi$ is a metric using $\Phi(z) = \exp(-z)$.  
\end{theorem}
\begin{proof}
    By embedding in $\R^m$ and using the Euclidean distance, we automatically inherit symmetry, triangle inequality, and that the distance is at least $0$ (non-negativity), making it a \emph{pseudometric}.  What remains is to show that $\|V_1 - V_2\|=0$ if and only if $X_1 = R(X_2)$ under some rotation $R \in \Rm$ or $\RoC$ operation.    This follows in two steps.  First, as argued above setting $f(j) = X(\theta_j)$ is an isometric mapping for $X \in \Sh^m_\star$ into $\F^m_{\neq}$.  Then by Theorem \ref{thm:main-invariance} two functions $f_1, f_2 \in \F^m_{[\neq]}$ satisfy $[f_1] = [f_2]$ (up to rotation and $\RoC$), if and only if their corresponding sketch vectors $V_1 = V_2$ are equal. 
\end{proof}

\paragraph*{Star Distance.}
For $X \in \Sh_\star$, and angle $\theta \in [0,2 \pi)$ let $X(\theta) = a_\theta$ where for ray $r_\theta$ the intersection $I(X \cap r_\theta) = [0,a_\theta)$.  
Now consider two $X_1, X_2 \in \Sh^\star$, and define the following \emph{star distance}:
\[
\D_\star(X_1, X_2) = \max_{\theta \in [0,2\pi)} |X_1(\theta) - X_2(\theta)|.
\]
The star distance between two shapes under rotations, what we call the \emph{R-star distance}:  
\[
\D_{R\star}(X_1, X_2) = \min_{R \in \mathbb{SO}(2)} \D_\star(X_1, R(X_2)).
\]

\paragraph*{Discrete Approximation.}
We can show that $\D_\star$ is approximately preserved as we replace a $\eta$-Lipschitz shape $X \in \Sh^{\star,\eta}$ with some $X' \in \Sh_{\star,\eta}^m$.  In particular choose a set $\mathcal{T} = \{\theta_1, \theta_2, \ldots, \theta_m\}$  so $\theta_j = j \frac{2\pi}{m} + \frac{\pi}{m}$.  That is each $r_{\theta_j}$ is the ray in the middle of the wedge of angles in $\Theta_j$.  
Define $X'$ as the \emph{$m$-discrete approximation of $X$} by $X'(\theta_j) = X(\theta_j)$; that is it takes the angular radius in the middle of each wedge $\Theta_j$. 

\begin{lemma} \label{lem:m-approx}
    For $X' \in \Sh^m_\star$ the $m$-discrete approximation of $X \in \Sh_{\star,\eta}$ then $\D_\star(X,X') \leq \pi \eta/m$.
\end{lemma}
\begin{proof}
    For each $\theta \in \Theta_j$, we have $|X(\theta) - X(\theta_j)| \leq \eta \frac{\pi}{m}$.  Since for any $\theta \in \Theta_j$ we have $X'(\theta) = X(\theta_j)$, the bound follows.  
\end{proof}

\begin{lemma} \label{lem:Dstar-L}
    For two $X_1, X_2 \in \Sh_\star$ if $\D_{R\star}(X_1, X_2) \leq \eps$, then the corresponding rotational invariant sketches $\psi_1, \psi_2$ satisfy $\|\psi_1 - \psi_2\|_{L_2} \leq (2e) \eps$.   
\end{lemma}
\begin{proof}
   Without loss of generality let the optimal rotation for $\D_{R\star}$ be the identity rotation $R(X_2) = X_2$, so now $\D_\star(X_1, X_2) \leq \eps$.  
   Then we know that $\max_\theta |X_1(\theta) - X_2(\theta)| \leq \eps$.  So then the max difference over all choices of $\theta,\alpha \in [0,2\pi)$ in 
   \[
   |(X_1(\theta) - X_1(\theta+\alpha)) - (X_2(\theta) - X_2(\theta+\alpha)| 
   \leq 
   |X_1(\theta) -  X_2(\theta)| + |X_1(\theta+\alpha) - X_2(\theta+\alpha)| 
   \leq 2 \eps.
   \]  
   Then we observe that $\Phi(z) = \exp(-z)$ is $e$-Lipschitz for $z \in [-1,1]$, and each of $f_1, f_2$ have range $(0,1]$, so $X_1(\theta) - X_1(\theta+\alpha) \in [-1,1]$ and same for $X_2$.  Thus for all $\alpha$
   \begin{align*}
   |\psi_1(\alpha) - \psi_2(\alpha)| 
   & \leq
   \frac{1}{2 \pi} \int_0^{2\pi} |\Phi(X_1(\theta) - X_1(\theta+\alpha)) - \Phi(X_2(\theta) - X_2(\theta+\alpha))| \dir \theta
   \\ &\leq
   \max_{\theta \in [0,2\pi]} |\Phi(X_1(\theta) - X_1(\theta+\alpha)) - \Phi(X_2(\theta) - X_2(\theta+\alpha))|
   \\ &\leq
   |\Phi(z) - \Phi(z+2\eps)|
   \leq
   e \cdot 2\eps.
   \end{align*}
   As a result 
   \[
   \|\psi_1 - \psi_2\|^2_{L_2} 
   = 
   \frac{1}{2\pi}\int_0^{2\pi} (\psi_1(\alpha) - \psi_2(\alpha))^2 \dir \alpha
   \leq
   \max_{\alpha \in [0,2\pi)} (\psi_1(\alpha) - \psi_2(\alpha))^2
   \leq (2e \eps)^2.  \qedhere
   \]
\end{proof}

\begin{theorem}
    For two shapes $X_1, X_2 \in \Sh_{\star,\eta}$ and $\Phi(z) = \exp(-z)$ we can compute a distance $D$ so that $|D - \D_\Phi(X_1,X_2)| \leq \eps$ in $O((1/\eta \eps) \log (1/\eta \eps))$ time.  
\end{theorem}
\begin{proof}
    Set $m = \frac{4e}{\pi \eta \eps}$ so that when using $X'$ an $m$ discrete approximation of $X \in \Sh_{\star, \eta}$ that $\D_\star(X,X') \leq \pi \eta / m = \eps/(4e)$, via Lemma \ref{lem:m-approx}.  
    This step is the only approximation, and takes $O(m)$ time for each shape $X_1, X_2$.  Then we represent them as functions $f_1, f_2 \in \F^m$, and convert them into sketches $V_1, V_2$ in $O(m \log m)$ time, via Theorem \ref{thm:FFT-sketch}.  By Lemma \ref{lem:Dstar-L} and $\D_\star(X_1,X'_1) \leq \eps/(4e)$ and $\D_\star(X_2,X'_2) \leq \eps/(4e)$ then by triangle inequality we have $\|V_1 - V_2\| \leq 2 \cdot (2e) \frac{\eps}{4e} = \eps$ as desired.   
    Computing $\|V_1 - V_2\|$ takes $O(m)$ time, so the total time is $O(m \log m) = O((1/\eta \eps) \log(1/\eta \eps))$.  
\end{proof}

Note that for two standardized star-shapes $X_1,X_2 \in \Sh_{\star}$ if $X_2$ is an $\eps$-Hausdorff distortion of $X_1$ (e.g., $\D_H(X_1, X_2) \leq \eps$), then this implies that $\D_\star(X_1,X_2) \leq \eps$, and hence by Lemma \ref{lem:Dstar-L}, their associated sketched functions $\psi_1, \psi_2$ satisfy $\|\psi_1 - \psi_2\|_{L_2} \leq (2e) \eps$.  This holds for shapes which are rotations of each other, with $\eps$-Hausdorff perturbations.  Thus, in this sense, in the rotation-shape space the sketch mapping is $2e$-Lipschitz with respect to Hausdorff distance.

%

\section{Experimental Validation}\footnote{Code and data to reproduce all experiments in this section are publicly available at \url{https://github.com/Hamidmath/Representation}; the star-shaped pipeline used in this paper is on the default branch.}
\label{sec:exp}

We assess the usefulness of our embedding and metric in two ways.  The first is how well our method can cluster data, and the second is in $k$-nearest neighbor search.
We focus here on the SQUID dataset~\cite{mokhtarian2013curvature,nasreddine2015shape}, which has $1100$ fish outlines; see examples in Figure~\ref{fig:raw_data}.  We treat the interior of each outline as a set $X \in \Sh$ after the standardization process.  While some shapes are star-shaped, to use our star-shaped mechanism, we first convert each shape into a star-shaped $X \in \Sh_{\star}^m$ by fitting each $X(\theta_j)$ value as the maximum point in the original shape along rays $r_\theta$ with $\theta \in \Theta_j$.  We then consider this set $\mathcal{X} \subset \Sh_m^\star$ as input.  When not specified, we use $m=128$.

\begin{figure}[H]
  \centering
  \includegraphics[width=\textwidth]{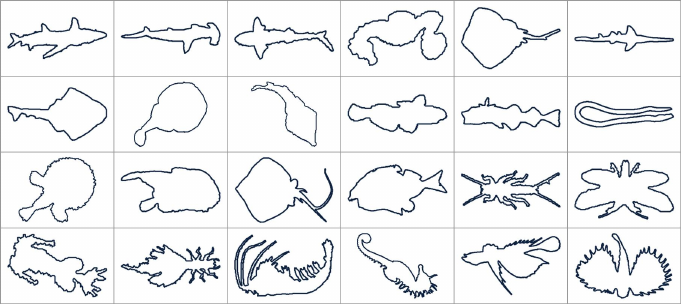}
  \caption{2D boundary points for sample fish shapes from the SQUID dataset.}
  \label{fig:raw_data}
\end{figure}

\paragraph*{Clustering.}
To assess the robustness of the proposed sketch vector with respect to rotation, and in particular how it interacts with the discretization level, we conducted the following experiment on a subset of star-shaped objects from the SQUID dataset.
We randomly select ten distinct shapes and standardize them.  For each shape, nine additional rotated copies were generated using uniformly random rotation angles in $[0, 2\pi)$, resulting in a collection $\mathcal{X}$ of $100$ objects in total.
We then convert each of these $100$ shapes to $\Sh_\star^m$ using our star-discretization process to get $\mathcal{X}^\star$, and each shape $X^\star_i \in \mathcal{X}^\star$ is converted to its vectorized representation $V_i \in \R^m$ to produce a set $\mathcal{V} \subset \R^m$.
We run a standard $k$-means clustering algorithm (with $k=10$) on $\mathcal{V}$, and measure the \emph{accuracy} of the clustering using the original object as a key: the fraction of shapes that are in the same cluster as the original shape they were rotated from.

The left panel of Figure~\ref{fig:exp-combined} illustrates the effectiveness of our representation.  It repeats the clustering experiment for a variety of values $m$ ranging from $8$ to $1024$ (plotted along the $x$-axis), reporting the mean accuracy with a shaded $\pm 1$ standard deviation band over $6$ random trials.  As expected, as $m$ increases we obtain a better representation of each shape: the average accuracy reaches about $0.95$ when $m$ hits roughly $100$, and is consistently above $0.99$ once $m > 250$.  It does not always reach $1.00$ due to occasional sampling of very similar fish shapes.

Note that when discretization is applied \textbf{before} the random rotation, the accuracy is $1.00$, since the signatures are rotation invariant.  So in addition to showing how easy it is to run a clustering algorithm with these signatures, this is a demonstration of the effectiveness of the star-shaped discretization process.

\paragraph*{$k$-NN Searching.}
We applied $5$-nearest-neighbor (5NN) search using the $\D_\Phi$ distance via signatures $V_i \in \R^m$ of the $1100$ vectorized objects $X_i$ in the full SQUID dataset, with $m=128$.
The right panel of Figure~\ref{fig:exp-combined} shows $3$ query fish (leftmost column, labeled ``Query''), and their $5$ nearest neighbors in descending order of closeness.  The returned fish are all similar in shape to the query but vary in their rotation, confirming the rotation invariance of the sketch.

\begin{figure}[H]
    \centering
    \includegraphics[width=\linewidth]{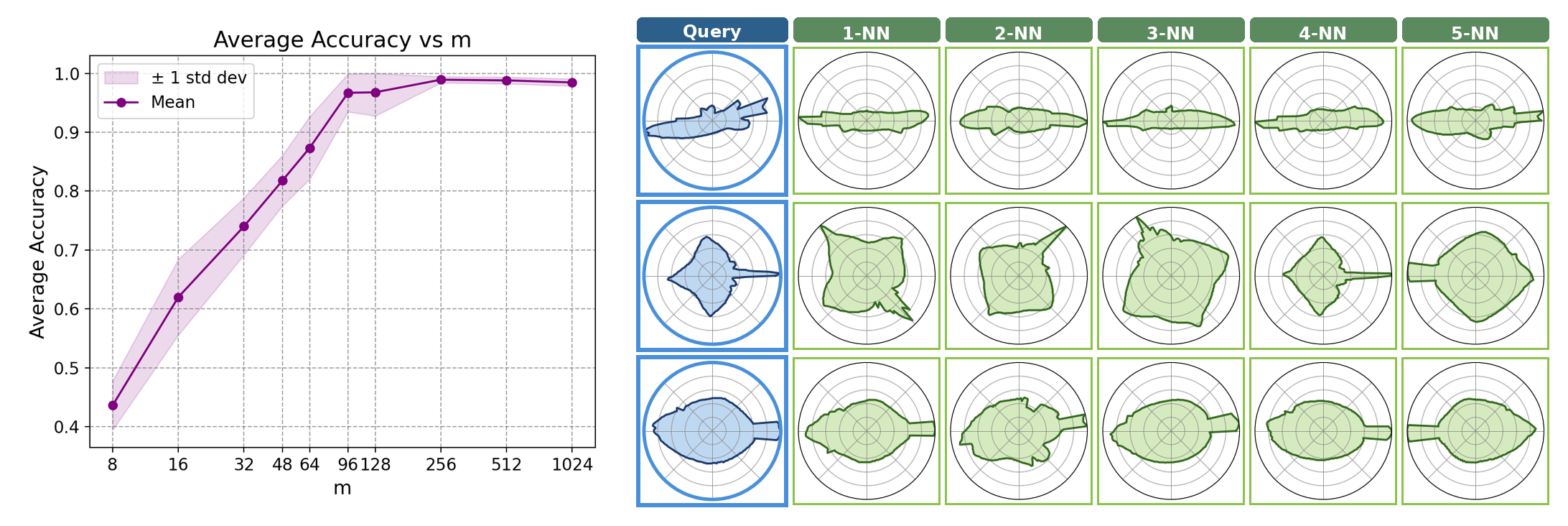}
    \caption{Left: clustering accuracy as a function of the discretization parameter $m$ ($10$ copies of $10$ shapes under random rotation, $k$-means with $k=10$; mean with $\pm 1$ standard deviation band over $6$ trials).  Right: nearest-neighbor queries (leftmost column, ``Query'') using their rotation-invariant signature with $m=128$; within each row, neighbors closer to the query appear on the left.}
    \label{fig:exp-combined}
\end{figure}

\bibliography{refs.bib}

@inproceedings{RahimiRecht2007,
  author = {A. Rahimi and B. Recht},
  title = {Random features for large-scale kernel machines},
  booktitle = {Advances in Neural Information Processing Systems (NIPS)},
  pages = {1177--1184},
  year = {2007}
}

@inproceedings{Mumford1993,
  author = {D. Mumford},
  title = {Pattern theory: the mathematics of perception},
  booktitle = {Proceedings of the International Congress of Mathematicians},
  publisher = {Birkhäuser},
  pages = {144--157},
  year = {1993}
}

@book{Srivastava2016,
  author = {A. Srivastava and E. Klassen},
  title = {Functional and Shape Data Analysis},
  publisher = {Springer},
  year = {2016}
}

@article{CooleyTukey1965,
  author = {J. W. Cooley and J. W. Tukey},
  title = {An algorithm for the machine calculation of complex Fourier series},
  journal = {Mathematics of Computation},
  volume = {19},
  number = {90},
  pages = {297--301},
  year = {1965}
}

@book{Oppenheim2009,
  author = {A. V. Oppenheim and R. W. Schafer},
  title = {Discrete-Time Signal Processing},
  publisher = {Prentice Hall},
  edition = {3rd},
  year = {2009}
}

@article{belongie2002,
  author = {S. Belongie and J. Malik and J. Puzicha},
  title = {Shape matching and object recognition using shape contexts},
  journal = {IEEE Transactions on Pattern Analysis and Machine Intelligence},
  volume = {24},
  number = {4},
  pages = {509--522},
  year = {2002}
}

@inproceedings{marr1977,
  author = {D. Marr and H. K. Nishihara},
  title = {Representation and recognition of the spatial organization of three-dimensional shapes},
  booktitle = {Proceedings of the Royal Society of London B: Biological Sciences},
  volume = {197},
  number = {1129},
  pages = {42--82},
  year = {1977}
}

@book{baker1975,
  author    = {A. Baker},
  title     = {Transcendental Number Theory},
  publisher = {Cambridge University Press},
  address   = {Cambridge, UK},
  year      = {1975}
}

@book{dryden98,
  author    = {Ian L. Dryden and Kanti Mardia},
  title     = {Statistical Shape Analysis},
  publisher = {John Wiley \& Sons},
  year      = {1998}
}

@inproceedings{agarwal2006bipartite,
  title={On bipartite matching under the RMS distance},
  author={Agarwal, Pankaj K and Phillips, Jeff M},
  booktitle={Proceedings of the Eighteenth Canadian Conference on Computational Geometry’, Kingston, Canada},
  pages={143--146},
  year={2006}
}

@inproceedings{ben2014minimum,
  title={Minimum partial-matching and Hausdorff RMS-distance under translation: combinatorics and algorithms},
  author={Ben-Avraham, Rinat and Henze, Matthias and Jaume, Rafel and Keszegh, Bal{\'a}zs and Raz, Orit E and Sharir, Micha and Tubis, Igor},
  booktitle={European Symposium on Algorithms},
  pages={100--111},
  year={2014},
  organization={Springer}
}

@InProceedings{ahsw-hdtpd-03,
  author =       {P. K. Agarwal and S.{Har-Peled} and M. Sharir and Y. Wang},
  title =        {Hausdorff Distance under Translation for Points, Disks,
                  and Balls},
  booktitle =    {Symposium on Computational Geometry},
  pages =        {282--291},
  year =         2003,
}

@inproceedings{huttenlocher1992dynamic,
  title={On dynamic Voronoi diagrams and the minimum Hausdorff distance for point sets under Euclidean motion in the plane},
  author={Huttenlocher, Daniel P and Kedem, Klara and Kleinberg, Jon M},
  booktitle={Proceedings of the eighth annual symposium on Computational geometry},
  pages={110--119},
  year={1992}
}

@article{arkin1991efficiently,
  title={An efficiently computable metric for comparing polygonal shapes},
  author={Arkin, Esther M and Chew, L Paul and Huttenlocher, Daniel P and Kedem, Klara and Mitchell, Joseph S},
  year={1991}
}

@article{osada2002shape,
  title={Shape distributions},
  author={Osada, Robert and Funkhouser, Thomas and Chazelle, Bernard and Dobkin, David},
  journal={ACM Transactions on Graphics (TOG)},
  volume={21},
  number={4},
  pages={807--832},
  year={2002},
  publisher={ACM New York, NY, USA}
}

@article{veltkamp2001state,
  title={State of the art in shape matching},
  author={Veltkamp, Remco C and Hagedoorn, Michiel},
  journal={Principles of visual information retrieval},
  pages={87--119},
  year={2001},
  publisher={Springer}
}

@article{cheng2013shape,
  title={Shape matching under rigid motion},
  author={Cheng, Siu-Wing and Lam, Chi-Kit},
  journal={Computational geometry},
  volume={46},
  number={6},
  pages={591--603},
  year={2013},
  publisher={Elsevier}
}

@article{chew1997geometric,
  title={Geometric pattern matching under Euclidean motion},
  author={Chew, L Paul and Goodrich, Michael T and Huttenlocher, Daniel P and Kedem, Klara and Kleinberg, Jon M and Kravets, Dina},
  journal={Computational Geometry},
  volume={7},
  number={1-2},
  pages={113--124},
  year={1997},
  publisher={Elsevier}
}

@inproceedings{chan2024convex,
  title={Convex Polygon Containment: Improving Quadratic to Near Linear Time},
  author={Chan, Timothy M and Hair, Isaac M},
  booktitle={40th International Symposium on Computational Geometry (SoCG 2024)},
  pages={34--1},
  year={2024},
  organization={Schloss Dagstuhl--Leibniz-Zentrum f{\"u}r Informatik}
}

@inproceedings{chan2025linear,
  title={A Linear Time Algorithm for the Maximum Overlap of Two Convex Polygons Under Translation},
  author={Chan, Timothy M and Hair, Isaac M},
  booktitle={41st International Symposium on Computational Geometry (SoCG 2025)},
  pages={31--1},
  year={2025},
  organization={Schloss Dagstuhl--Leibniz-Zentrum f{\"u}r Informatik}
}

@article{agarwal1998largest,
  title={Largest placement of one convex polygon inside another},
  author={Agarwal, Pankaj K and Amenta, Nina and Sharir, Micha},
  journal={Discrete \& Computational Geometry},
  volume={19},
  number={1},
  pages={95--104},
  year={1998},
  publisher={Springer}
}

@article{memoli2012some,
  title={Some properties of Gromov--Hausdorff distances},
  author={M{\'e}moli, Facundo},
  journal={Discrete \& Computational Geometry},
  volume={48},
  number={2},
  pages={416--440},
  year={2012},
  publisher={Springer}
}

@inproceedings{memoli2008gromov,
  title={Gromov-Hausdorff distances in Euclidean spaces},
  author={M{\'e}moli, Facundo},
  booktitle={2008 IEEE Computer Society Conference on Computer Vision and Pattern Recognition Workshops},
  pages={1--8},
  year={2008},
  organization={IEEE}
}

@article{schmiedl2017computational,
  title={Computational aspects of the Gromov--Hausdorff distance and its application in non-rigid shape matching},
  author={Schmiedl, Felix},
  journal={Discrete \& Computational Geometry},
  volume={57},
  number={4},
  pages={854--880},
  year={2017},
  publisher={Springer}
}

@inproceedings{qi2017pointnet,
  title={Pointnet: Deep learning on point sets for 3d classification and segmentation},
  author={Qi, Charles R and Su, Hao and Mo, Kaichun and Guibas, Leonidas J},
  booktitle={Proceedings of the IEEE conference on computer vision and pattern recognition},
  pages={652--660},
  year={2017}
}

@article{qi2017pointnet++,
  title={Pointnet++: Deep hierarchical feature learning on point sets in a metric space},
  author={Qi, Charles Ruizhongtai and Yi, Li and Su, Hao and Guibas, Leonidas J},
  journal={Advances in neural information processing systems},
  volume={30},
  year={2017}
}

@inproceedings{kazhdan2003rotation,
  title={Rotation invariant spherical harmonic representation of 3 d shape descriptors},
  author={Kazhdan, Michael and Funkhouser, Thomas and Rusinkiewicz, Szymon},
  booktitle={Symposium on geometry processing},
  volume={6},
  pages={156--164},
  year={2003}
}

@article{xiao2023unsupervised,
  title={Unsupervised point cloud representation learning with deep neural networks: A survey},
  author={Xiao, Aoran and Huang, Jiaxing and Guan, Dayan and Zhang, Xiaoqin and Lu, Shijian and Shao, Ling},
  journal={IEEE Transactions on Pattern Analysis and Machine Intelligence},
  volume={45},
  number={9},
  pages={11321--11339},
  year={2023},
  publisher={IEEE}
}

@inproceedings{phillips2020sketched,
  title={Sketched MinDist},
  author={Phillips, Jeff M and Tang, Pingfan},
  booktitle={36th International Symposium on Computational Geometry (SoCG 2020)},
  pages={63--1},
  year={2020},
  organization={Schloss Dagstuhl--Leibniz-Zentrum f{\"u}r Informatik}
}

@book{kendall2009shape,
  title={Shape and shape theory},
  author={Kendall, David George and Barden, Dennis and Carne, Thomas K and Le, Huiling},
  year={2009},
  publisher={John Wiley \& Sons}
}

@article{aumuller2020ann,
  title={ANN-Benchmarks: A benchmarking tool for approximate nearest neighbor algorithms},
  author={Aum{\"u}ller, Martin and Bernhardsson, Erik and Faithfull, Alexander},
  journal={Information Systems},
  volume={87},
  pages={101374},
  year={2020},
  publisher={Elsevier}
}

@article{malkov2018efficient,
  title={Efficient and robust approximate nearest neighbor search using hierarchical navigable small world graphs},
  author={Malkov, Yu A and Yashunin, Dmitry A},
  journal={IEEE transactions on pattern analysis and machine intelligence},
  volume={42},
  number={4},
  pages={824--836},
  year={2018},
  publisher={IEEE}
}

@article{douze2025faiss,
  title={The faiss library},
  author={Douze, Matthijs and Guzhva, Alexandr and Deng, Chengqi and Johnson, Jeff and Szilvasy, Gergely and Mazar{\'e}, Pierre-Emmanuel and Lomeli, Maria and Hosseini, Lucas and J{\'e}gou, Herv{\'e}},
  journal={IEEE Transactions on Big Data},
  year={2025},
  publisher={IEEE}
}

@article{andoni2015practical,
  title={Practical and optimal LSH for angular distance},
  author={Andoni, Alexandr and Indyk, Piotr and Laarhoven, Thijs and Razenshteyn, Ilya and Schmidt, Ludwig},
  journal={Advances in neural information processing systems},
  volume={28},
  year={2015}
}

@article{hofmann2008kernel,
  title={Kernel Methods in Machine Learning},
  author={Hofmann, Thomas and Sch{\"o}lkopf, Bernhard and Smola, Alexander J},
  journal={The Annals of Statistics},
  volume={36},
  number={3},
  pages={1171--1220},
  year={2008}
}

@article{horn1987closed,
  title={Closed-form solution of absolute orientation using unit quaternions},
  author={Horn, Berthold KP},
  journal={Journal of the optical society of America A},
  volume={4},
  number={4},
  pages={629--642},
  year={1987},
  publisher={Optical Society of America}
}

@article{arun1987least,
  title={Least-squares fitting of two 3-D point sets},
  author={Arun, K Somani and Huang, Thomas S and Blostein, Steven D},
  journal={IEEE Transactions on pattern analysis and machine intelligence},
  number={5},
  pages={698--700},
  year={1987},
  publisher={IEEE}
}

@inproceedings{faugeras1983,
  author = {O.D Faugeras and M. Hebert},
  title={A {3-D} recognition and positioning algorithm using geometrical matching between primitive surfaces},
  booktitle={Proceedings of the Eighth international joint conference on Artificial intelligence-Volume 2},
  pages={996--1002},
  year={1983}
}

@article{hanson1981analysis,
  title={Analysis of measurements based on the singular value decomposition},
  author={Hanson, Richard J and Norris, Michael J},
  journal={SIAM Journal on Scientific and Statistical Computing},
  volume={2},
  number={3},
  pages={363--373},
  year={1981},
  publisher={SIAM}
}

@article{matouvsek2008variants,
  title={On variants of the Johnson--Lindenstrauss lemma},
  author={Matou{\v{s}}ek, Ji{\v{r}}{\'\i}},
  journal={Random Structures \& Algorithms},
  volume={33},
  number={2},
  pages={142--156},
  year={2008},
  publisher={Wiley Online Library}
}

@article{kane2014sparser,
  title={Sparser johnson-lindenstrauss transforms},
  author={Kane, Daniel M and Nelson, Jelani},
  journal={Journal of the ACM (JACM)},
  volume={61},
  number={1},
  pages={1--23},
  year={2014},
  publisher={ACM New York, NY, USA}
}

@book{mokhtarian2013curvature,
  title={Curvature scale space representation: theory, applications, and MPEG-7 standardization},
  author={Mokhtarian, Farzin and Bober, Miroslaw},
  volume={25},
  year={2013},
  publisher={Springer Science \& Business Media}
}

@inproceedings{nasreddine2015shape,
  title={Shape-based fish recognition via shape space},
  author={Nasreddine, Kamal and Benzinou, Abdesslam},
  booktitle={2015 23rd European Signal Processing Conference (EUSIPCO)},
  pages={145--149},
  year={2015},
  organization={IEEE}
}

@article{arya2025sheaf,
  title={A sheaf-theoretic construction of shape space},
  author={Arya, Shreya and Curry, Justin and Mukherjee, Sayan},
  journal={Foundations of Computational Mathematics},
  volume={25},
  number={3},
  pages={813--863},
  year={2025},
  publisher={Springer}
}

@article{curry2022many,
  title={How many directions determine a shape and other sufficiency results for two topological transforms},
  author={Curry, Justin and Mukherjee, Sayan and Turner, Katharine},
  journal={Transactions of the American Mathematical Society, Series B},
  volume={9},
  number={32},
  pages={1006--1043},
  year={2022}
}

\newpage
\appendix



\section{More Formal Definitions}
\label{app:defn}

\begin{definition}[Discrete circle and modular arithmetic]\label{def:discrete_circle}
Fix an integer \(m \ge 1\).  
We represent the discrete circle by the index set \(\mathbb Z_m := \{0,1,\dots,n-1\}\) with modular addition
\[
a +_m b := (a+b) \bmod m.
\]
For each index \(k \in \mathbb Z_m\), define the discrete angle
\[
\theta_k := \frac{2\pi k}{m}.
\]
The additive inverse of \(\theta_k\) is
\[
-\theta_k := \theta_{-k}, \qquad \text{where } -k := (-k)\bmod m = 
\begin{cases}
0, & k=0,\\
m-k, & 1\le k\le m-1.
\end{cases}
\]
All index operations are taken modulo \(m\).
\end{definition}

\begin{definition}[Discrete signature]\label{def:discrete_signature}
Let \(f = (f_0, f_1, \dots, f_{m-1}) \in \mathbb R^m\) and let \(\Phi:\mathbb R\to\mathbb R\) be any fixed function (e.g., \(\Phi(x)=e^{-|x|}\) or \(\Phi(x)=e^{-x^2}\)).  
The \emph{discrete signature} of \(f\) is the function \(V_f:\mathbb Z_m \to \mathbb R\) defined by
\[
V_f[r] := \frac{1}{m}\sum_{j=0}^{m-1} \Phi\!\big(f_{j+_m r} - f_j\big), \qquad r \in \mathbb Z_m.
\]
\end{definition}

\begin{theorem}[Discrete rotation and shift invariance]\label{thm:rotation_shift_invariance}
Let \(f=(f_0,\dots,f_{m-1})\in\mathbb R^m\).  
For any rotation \(b\in\mathbb Z_m\) and constant \(c\in\mathbb R\), define
\[
g_j := f_{j+_m b} + c, \qquad j\in\mathbb Z_m.
\]
Then for all \(r\in\mathbb Z_m\),
\[
V_g[r] = V_f[r].
\]
\end{theorem}

\begin{proof}
\[
\begin{aligned}
V_g[r]
&= \frac{1}{m}\sum_{j=0}^{m-1} \Phi\!\big(g_{j+_m r} - g_j\big)
 = \frac{1}{m}\sum_{j=0}^{m-1} \Phi\!\big(f_{j+_m r+_m b} + c - (f_{j+_m b} + c)\big)\\
&= \frac{1}{m}\sum_{j=0}^{m-1} \Phi\!\big(f_{j+_m r+_m b} - f_{j+_m b}\big)
 = \frac{1}{m}\sum_{i=0}^{m-1} \Phi\!\big(f_{i+_m r} - f_i\big) = V_f[r],
\end{aligned}
\]
where we reindexed \(i = j+_m b\).  
\end{proof}

\begin{theorem}[Discrete involution, rotation, and shift invariance]\label{thm:involution_invariance}
Let \(f\in\mathbb R^m\), and fix \(b\in\mathbb Z_m\), \(c\in\mathbb R\).  
Define
\[
g_j := -\,f_{-j+_m b} + c, \qquad j\in\mathbb Z_m.
\]
Then for all \(r\in\mathbb Z_m\),
\[
V_g[r] = V_f[r].
\]
\end{theorem}

\begin{proof}
\[
\begin{aligned}
V_g[r]
&= \frac{1}{m}\sum_{j=0}^{m-1}\Phi\!\big(g_j - g_{j+_m r}\big)
 = \frac{1}{m}\sum_{j=0}^{m-1}\Phi\!\big(-f_{-j+_m b}+c - (-f_{-(j+_m r)+_m b}+c)\big)\\
&= \frac{1}{m}\sum_{j=0}^{m-1}\Phi\!\big(f_{-j-_m r+_m b} - f_{-j+_m b}\big)
 = \frac{1}{m}\sum_{i=0}^{m-1}\Phi\!\big(f_{i-_m r} - f_i\big)
 = \frac{1}{m}\sum_{i=0}^{m-1}\Phi\!\big(f_i - f_{i+_m r}\big)
 = V_f[r],
\end{aligned}
\]
where \(i=-j+_m b\) and indices are taken modulo \(m\).
\end{proof}

\begin{definition}[Discrete Invariant Signature Kernel]
\label{def:discrete_invariant_kernel}
Let $f,g\in\mathbb R^m$ be two discrete functions on the circle $\mathbb Z_m$, 
and let $V_f,S_g:\mathbb Z_m\to\mathbb R$ denote their discrete signature functions defined by
\[
V_f(r) := \frac{1}{m}\sum_{j=0}^{m-1}\Phi\bigl(f_j - f_{j+_m r}\bigr),
\qquad
V_g(r) := \frac{1}{m}\sum_{j=0}^{m-1}\Phi\bigl(g_j - g_{j+_m r}\bigr).
\]
The \emph{Discrete Invariant Signature Kernel} between $f$ and $g$ is defined as the discrete inner product
\[
k(f,g)
\;:=\;
\langle V_f, V_g\rangle_{\ell^2(\mathbb Z_m)}
\;=\;
\frac{1}{m}\sum_{r=0}^{m-1} V_f(r)\,\overline{V_g(r)}.
\]
This kernel quantifies the total similarity between $f$ and $g$ across all discrete shifts on the circle $\mathbb Z_m$.
\end{definition}

\begin{remark}[Invariance]
If $g_j = f_{j+_m b}+c$ or $g_j = -\,f_{-j+_m b}+c$ for some $b\in\mathbb Z_m$ and $c\in\mathbb R$, 
then $V_g(r)=V_f(r)$ for all $r\in\mathbb Z_m$.  
Consequently, the discrete invariant signature kernel is unchanged:
\[
k(f,g) = k(g,f) = k(f,f) = k(g,g).
\]
\end{remark}


\section{Invariance of Permutation Functions}
\label{app:permutations}

\subsection{Lag-Homomorphism, and other Definitions}
\label{app:lag-homomorph}

\begin{definition}
    Consider the additive group $\cms$. 
    Every function $f \in \ms$ can be represented by the sequence 
    \(
        f = (f(0), f(1), \ldots, f(m-1)).
    \)
    By periodicity, $f(m) = f(0)$ and more generally $f(-j) = f(m-j)$ for all integers $j$.
\end{definition}

\begin{definition}
    For $0 \le x \le m-1$, the right and left rotations of $f$ by $x$ are
    \[
        \Rot_x^r(f)(i)=f(i-x), \qquad 
        \Rot_x^l(f)(i)=f(i+x).
    \]
    These definitions extend naturally to any integer $x$ modulo $m$.
\end{definition}

\begin{definition}
    For $f \in \ms$ and a constant $c \in \mathbb{R}$, the \emph{Reverse-of-Complement} of $f$ with respect to $c$ is defined by
    \(
        \RoC_c(f)(i) = c - f(-i).
    \)
    When $c = 0$, we simply write $\RoC(f)$.
\end{definition}
\begin{definition}
For each shift $k \in \{0,\dots,m-1\}$ and any $f \in \ms$, define the multiset of discrete differences
\[
    \Delta_f(k) := \{\, f(j+k) - f(j) : j = 0,\dots,m-1 \,\}.
\]
If $f, g \in \ms$ satisfy $\Delta_f(k) = \Delta_g(k)$ for all $k\in\cms$, then we say that $f$ and $g$ are \emph{lag-homometric} sequences. 
\end{definition}

\begin{theorem}\label{thm:roc_preserves_lag_homometry}
Let $f,g \in \ms$ be lag-homometric. Then for any $c \in \mathbb{R}$,
the sequences $f$ and $\RoC_c(g)$ are lag-homometric.
\end{theorem}

\begin{proof}
Fix $k \in \cms$.  
Since $\RoC_c(g)(j) = c - g(-j)$, we compute
\[
    \RoC_c(g)(j+k)-\RoC_c(g)(j)
    = \big(c - g(-j-k)\big) - \big(c - g(-j)\big)
    = g(-j) - g(-j-k).
\]
Setting $i = -j-k$ (a permutation of indices modulo $n$) gives
\begin{align*}
    &g(-j) - g(-j-k) = g(i+k) - g(i),\\ &\Longrightarrow
    \Delta_{\RoC_c(g)}(k)
    = \{\, \RoC_c(g)(j+k)-\RoC_c(g)(j) : j \,\}
    = \{\, g(i+k)-g(i) : i \,\}
    = \Delta_g(k).
\end{align*}
Since $f$ and $g$ are lag-homometric, $\Delta_f(k)=\Delta_g(k)$ for all $k$, and therefore
$\Delta_f(k)=\Delta_{\RoC_c(g)}(k)$ for all $k$.  
Thus $f$ and $\RoC_c(g)$ are lag-homometric.
\end{proof}

\begin{theorem}\label{thm:rotation_preserves_lag_homometry}
Let $f,g \in \ms$ be lag-homometric. Then for any $x \in \mathbb{Z}$,
the sequences $f$ and $\Rot_x(g)$ are lag-homometric.
\end{theorem}

\begin{proof}
Fix $k \in \cms$.  
Since $\Rot_x^r(g)(j) = g(j-x)$, we compute
\[
    \Rot_x^r(g)(j+k) - \Rot_x^r(g)(j)
    = g(j+k-x) - g(j-x).
\]
Setting $i = j-x$ (again a permutation of indices modulo $n$) yields
\begin{align*}
    &g(j+k-x) - g(j-x) = g(i+k) - g(i),\\ \Longrightarrow
    &\Delta_{\Rot_x(g)}(k)
    = \{\,\Rot_x^r(g)(j+k)-\Rot_x^r(g)(j) : j \,\}
    = \{\, g(i+k)-g(i) : i \,\}
    = \Delta_g(k).
\end{align*}
Since $f$ and $g$ are lag-homometric, $\Delta_f(k)=\Delta_g(k)$ for all $k$, and therefore
$\Delta_f(k)=\Delta_{\Rot_x(g)}(k)$ for all $k$.  
Thus $f$ and $\Rot_x^r(g)$ are lag-homometric.
\end{proof}

\begin{definition}
    For $f \in \ms$ we define its equivalence class
    \[
        [f] \;:=\; 
        \{\, g \in \ms\;|\; 
        g = \Rot_m^r(\RoC_c(f)) \ \text{or}\ g = \Rot_m^r(f)
        \text{ for some } c \in \mathbb{R},\ m \in \cms \,\}.
    \]
\end{definition}

\begin{definition}
    Let $f \in \ms$ and let $\Phi$ be a fixed feature function.
    The \emph{discrete signature} of $f$ is the function $V_f \in \ms$ defined by
    \[
        V_f(k)
        = \frac{2\pi}{m} \sum_{j=0}^{m-1} \Phi\!\big( f(j) - f(j+k) \big),
        \qquad k \in \cms.
    \]
\end{definition}

\begin{remark}
    Let $\phi : [0,2\pi] \to \mathbb{R}$ be a function that is constant on each interval
    \[
        \Big[\,\tfrac{i\,2\pi}{m},\ \tfrac{(i+1)\,2\pi}{m}\,\Big],
        \qquad i \in \cms,
    \]
    and denote this constant value by $t_i$. Then
    \[
        \int_0^{2\pi} \phi(t)\, dt
        = \sum_{i=0}^{m-1} \int_{\frac{i2\pi}{m}}^{\frac{(i+1)2\pi}{m}} t_i\, dt
        = \sum_{i=0}^{m-1} t_i \,\frac{2\pi}{m}
        = \frac{2\pi}{m} \sum_{i=0}^{m-1} t_i.
    \]
    Hence the integral of such a uniformly spaced step function is exactly proportional to
    the sum of its step values. This shows that our discrete signature coincides with the
    continuous signature on the space of these step functions.
\end{remark}

\begin{theorem}
    If $f, g \in \ms$ are lag-homometric (in particular, if $g \in [f]$), then
    \[
        V_f(k) = V_g(k) \qquad\text{for all } k \in \cms.
    \]
\end{theorem}

\begin{proof}
    For each $k \in \cms$ we may rewrite
    \[
        V_f(k)
        = \frac{2\pi}{m} \sum_{j=0}^{m-1} \Phi\!\big(f(j+k)-f(j)\big)
        = \frac{2\pi}{m} \sum_{x \in \Delta_f(k)} \Phi(x),
    \]
    since the sum enumerates precisely the elements of the multiset $\Delta_f(k)$.
Because $f$ and $g$ are lag-homometric, Theorems~\ref{thm:roc_preserves_lag_homometry}
    and \ref{thm:rotation_preserves_lag_homometry} together ensure that
    \(
        \Delta_f(k) = \Delta_g(k) \text{ for all } k \in \cms.
    \)
    Therefore,
    \[
        V_g(k)
        = \frac{2\pi}{m} \sum_{x \in \Delta_g(k)} \Phi(x)
        = \frac{2\pi}{m} \sum_{x \in \Delta_f(k)} \Phi(x)
        = V_f(k).
    \]
\end{proof}

\begin{lemma}[Injectivity of Integer Exponential Sums]\label{lem:expo_multisets}
Let $m \geq 1$ and let $\lambda \ne 0$ be an algebraic real number. 
Fix an integer $p \ge 0$.  
Suppose 
\(
V=\{v_0, v_1, \ldots, v_{m-1}\}, 
\  
W=\{w_0, w_1, \ldots, w_{m-1}\}
\)
are two multisets of rational numbers satisfying
\[
\sum_{j=0}^{m-1} e^{\lambda\, v_j^{\,2p+1}}
    = \sum_{j=0}^{m-1} e^{\lambda\, w_j^{\,2p+1}}.
\]
Then \(V = W\) as multisets.
\end{lemma}

\begin{proof}
For each $k \in \mathbb{Q}$ define
\[
m_k = \big|\{\, i \mid v_i = k \,\}\big|,
\qquad
n_k = \big|\{\, i \mid w_i = k \,\}\big|.
\]
Then $m_k, n_k \in \mathbb{Z}_{\ge 0}$, and only finitely many of them are nonzero.

In terms of these multiplicities, the assumptions become
\[
\sum_{k \in \mathbb{Q}} m_k = n = \sum_{k \in \mathbb{Q}} n_k,
\qquad
\sum_{k \in \mathbb{Q}} m_k\, e^{\lambda k^{\,2p+1}}
    = \sum_{k \in \mathbb{Q}} n_k\, e^{\lambda k^{\,2p+1}}.
\]

Define $d_k := m_k - n_k$ for each $k \in \mathbb{Q}$. Then $d_k \in \mathbb{Z}$, and
\[
\sum_{k \in \mathbb{Q}} d_k\, e^{\lambda k^{\,2p+1}} = 0.
\]

Let
\[
S = \{\, k \in \mathbb{Q} : d_k \neq 0 \,\}
\]
be the support of $\{d_k\}$.  
If $S = \emptyset$, we are done.  
Assume $S \neq \emptyset$.  
Then
\[
\sum_{k \in S} d_k\, e^{\lambda k^{\,2p+1}} = 0.
\]

Now invoke the classical Lindemann--Weierstrass theorem 
(\cite[Theorem~1.4]{baker1975}), which asserts that if 
$\alpha_1,\ldots,\alpha_m$ are distinct algebraic numbers, then
$e^{\alpha_1},\ldots,e^{\alpha_m}$ are linearly independent, 
in particular linearly independent over~$\mathbb{Q}$.

Apply this theorem to the distinct algebraic numbers 
\[
\alpha_i = \lambda\, k_i^{\,2p+1},
\qquad S = \{k_1,\ldots,k_m\}.
\]
Then
\[
\sum_{i=1}^m d_{k_i}\, e^{\lambda k_i^{\,2p+1}} = 0,
\qquad d_{k_i} \in \mathbb{Q},
\]
implies $d_{k_i} = 0$ for all $i$ by linear independence.  
This contradicts the definition of $S$.  
Hence $S = \emptyset$.

Therefore $d_k = 0$ for all $k \in \mathbb{Q}$, so $m_k = n_k$ for all $k$.  
Thus $V$ and $W$ have identical multiplicities, and therefore
\(
V = W
\)
as multisets.
\end{proof}

\begin{corollary}
    If $f, g\in\ms$ and $V_f = V_g$ where $\Phi(x) = e^{\lambda x^{2p+1}}$ for a fix $ p\in\mathbb{Z}$, then they are lag-homometric.
\end{corollary}
\begin{proof}
Fix $k\in\cms$. From $V_f(k) = V_g(k)$ we have
\begin{align*}
&\frac{2\pi}{m} \sum_{j=0}^{m-1} \exp \lambda(f(j) - f(j+k))^{2p+1} = \frac{2\pi}{m} \sum_{j=0}^{m-1} \exp \lambda(g(j) - g(j+k))^{2p+1} \\
&\Longrightarrow
 \sum_{j=0}^{m-1} \exp \lambda(f(j)-  f(j+k))^{2p+1} = \sum_{j=0}^{m-1} \exp \lambda(g(j)-g(j+k))^{2p+1}\\
&\Longrightarrow
 \sum_{x\in\Delta_f(k)} \exp \lambda x^{2p+1} = \sum_{x\in\Delta_g(k)} \exp \lambda x^{2p+1}
\end{align*}
    By Lemma \ref{lem:expo_multisets}, we see that $\Delta_f(k) = \Delta_g(k)$
\end{proof}

\subsection{Uniqueness for Permutation Functions in $\F^m_\sigma$}
\label{app:permutations-inv}

\begin{definition}
    For $f\in \F^m$, the position function $\pi_f$ is defined by $\pi_f(f(i)) = i \bmod m$ for all $i\in\cms$.
\end{definition}

\begin{theorem}\label{thm:range_preserved}
Let $f,g \in \ms$.  
If $f$ and $g$ are lag-homometric and
\[
    \beta := \max_j f(j),\qquad 
    \alpha := \min_j f(j),\qquad
    \beta' := \max_j g(j),\qquad
    \alpha' := \min_j g(j),
\]
then $\beta-\alpha = \beta'-\alpha'$.  
Moreover, if $f$ and $g$ are injective, then
\[
    \pi_f(\beta) - \pi_f(\alpha)
    = 
    \pi_g(\beta') - \pi_g(\alpha')
    \pmod m.
\]
\end{theorem}

\begin{proof}
Since all indices are taken modulo $m$, every ordered pair $(p,q)$ appears as a
lagged difference: for any $p,q$ there exists $k$ with $q = p+k \pmod m$, and hence
\[
    f(q)-f(p) \in \Delta_f(k).
\]
Therefore,
\[
    \beta-\alpha
    = \max\{ d : d \in \Delta_f(k)\ \text{for some } k \}.
\]
Because $f$ and $g$ are lag-homometric, we have $\Delta_f(k)=\Delta_g(k)$ for all $k$, so the same maximal value appears in $\bigcup_k \Delta_g(k)$. Hence
\[
    \beta-\alpha = \beta'-\alpha'.
\]

For the positional statement, assume $f$ and $g$ are injective.  
Then the values $\beta,\alpha$ occur at unique indices $i,j$ with
\[
    f(i)=\beta,\qquad f(j)=\alpha,
\]
and similarly $\beta',\alpha'$ occur at unique indices $i',j'$ with
\[
    g(i')=\beta',\qquad g(j')=\alpha'.
\]

Since $f(i)-f(j)=\beta-\alpha$ is the maximal difference, this value occurs in
\(\Delta_f(k)\) only for the unique shift
\[
    k = i-j \pmod m.
\]
Similarly, the maximal difference $\beta'-\alpha'$ occurs in \(\Delta_g(k')\) only for
\[
    k' = i'-j' \pmod m.
\]
Because the two functions are lag-homometric, the location (shift) of the maximal difference must be the same; hence
\[
    k = k' \pmod m.
\]
Substituting the definitions of $k$ and $k'$ gives
\[
    i - j = i' - j' \pmod m,
\]
which in terms of position functions is exactly
\[
    \pi_f(\beta) - \pi_f(\alpha)
    =
    \pi_g(\beta') - \pi_g(\alpha')
    \pmod m.
\]

This completes the proof.
\end{proof}

\begin{corollary}[Anchoring \(1\) and \(m\) by rotation]\label{cor:anchor_1_m}
Let \(f,g\) be permutations of \(\{1,\dots,m\}\). If $f$ and $g$ are lag-homometric, then $\pi_f(m)- \pi_f(1) = \pi_g(m)- \pi_g(1) \bmod m$.
\end{corollary}
\begin{proof}
Apply Theorem~\ref{thm:range_preserved} with  
\(\beta=\beta'=m\) and \(\alpha=\alpha'=1\).
\end{proof}

\begin{example}
For $n=5$, if $f = (1, 3, 5, 4, 2)$ and $g = (3, 1, 4, 2, 5)$, then:
\begin{itemize}
\item In $f$: $\pi_f(1) = 0$, $\pi_f(5) = 2$, so cyclic distance is $2-0=2$ and $5-1 =4 \in \Delta_f(2)$
\item In $g$: $\pi_g(1) = 1$, $\pi_g(5) = 4$, so cyclic distance is $4-1=3$ and $5-1 =4 \in \Delta_g(3)$
\end{itemize}
Moreover,
\begin{align*}
4\notin \Delta_g(2) = \{3-1=2,\ 4-3=1,\ 2-4=-2,\ 5-2=3,\ 1-5=-4\}, \\
4\notin \Delta_f(3) = \{4-1=3,\ 2-3=-1,\ 1-5=-4,\ 3-4=-1,\ 5-2=3\}.
\end{align*}
Therefore, $V_f\neq V_g$.
\end{example}

\begin{remark}[Positions of the next extremes cannot be determined up to swap]\label{lem:place_2_and_m-1_up_to_swap}
Let $f,g$ be lag-homometric permutations of $\{1,\dots,m\}$ with
\(\pi_f(1)=\pi_g(1)\) and \(\pi_f(m)=\pi_g(m)\).
One cannot in general conclude that
\[
\{ \pi_g(2),\; \pi_g(m-1)\} = \{ \pi_f(2),\; \pi_f(m-1)\}.
\]
The following simple counterexample shows the failure of this claim:
\[
f=(1,2,4,3,5),\qquad g=(1,3,2,4,5).
\]
For these permutations
\[
\{\pi_f(2),\pi_f(4)\}=\{1,2\}\neq\{2,3\}=\{\pi_g(2),\pi_g(4)\}.
\]
Hence additional hypotheses are required if one wishes to determine the positions
of $2$ and $n-1$ up to swapping.
\end{remark}

\begin{lemma}\label{lem:secondcomponent}
Let $f,g$ be lag-homometric permutations of $\{1,\dots,m\}$.
Assume
\begin{enumerate}
  \item $\pi_f(1)=\pi_g(1)=0$,
  \item $\pi_f(m)=\pi_g(m)$,
 for every shift $k\in\mathbb Z_m$,
\end{enumerate}
Then
\[
\pi_f(2)\ne \pi_g(2)\Longrightarrow\pi_g(m-1) = \pi_f(m)- \pi_f(2) \pmod m.
\]
\end{lemma}
\begin{proof}
We prove this by contradiction. Suppose that
\[
\pi_g(m-1) \not= \pi_f(m)- \pi_f(2) \pmod m,
\]
and set $d = m-2$, $k = \pi_f(m) - \pi_f(2)$ and
\[
C_f(k,d) = \text{Multiplicity of } d \text{ at shift } k \text{ in } \Delta_f(k) = |\{d\ | d\in \Delta_f(k)\}|
\]
The only ordered pairs that can produce $d$ are $(2,m)$ and $(1, m-1)$. 
Now, let's compute $C_f(k,d)$ and $C_g(k,d)$, which must be equal by the 4th hypothesis.
\begin{itemize}
    \item \textbf{Computing $C_f(k,d)$:}\\
    Clearly, $(2,m)$ contributes 1 to $C_f(k,d)$ because $m-2\in \Delta_f(\pi_f(m) - \pi_f(2))$. On the other hand, the pair $(1,m-1)$ contributes 1 if $\pi_f(m-1) - \pi_f(1) = \pi_f(m-1)= k \pmod m$, and if not, it contributes 0 to $C_f(k,d)$. Therefore, $C_f(k,d)= 1$ or $C_f(k,d)=2$.
    \item \textbf{Computing $C_g(k,d)$:}\\
    We know that $m-2\in\Delta_g(\pi_g(m)- \pi_g(2))$. Now, if $\pi_g(m) - \pi_g(2) = k \pmod m$, then $\pi_f(m) -\pi_f(2) = \pi_g(m) - \pi_g(2) \pmod m$, and because $\pi_f(m) = \pi_g(m)$, we get $\pi_f(2) = \pi_g(2)$, which is a contradiction. Therefore, $\pi_g(m) - \pi_g(2) \not= k \pmod m$. Therefore, the pair $(2,m)$ contributes 0 to $C_g(k,d)$. On the other hand, the pair $(1,m-1)$ contributes 0 to $C_g(k,d)$ because $(m-1)-(1)\in \Delta_g(\pi_g(m-1) - \pi_g(1)) = \Delta_g(\pi_g(m-1))$ and $\pi_g(m-1)\not= \pi_f(m)- \pi_f(2) = k \pmod m$. This implies that $C_g(k,d) = 0$.
\end{itemize}
Therefore, $C_g(k,d)<C_f(k,d)$, which is a contradiction with being lag-homometric. Therefore, the assumption is false, and $\pi_g(m-1) = \pi_f(m)- \pi_f(2) \pmod m$.
\end{proof}

\begin{lemma}\label{lem:secondcomponentExtension}
Let $f,g$ be lag-homometric permutations of $\{1,\dots,m\}$.
Assume
\begin{enumerate}
  \item $\pi_f(1)=\pi_g(1)=0$,
  \item $\pi_f(m)=\pi_g(m)$,
  \item $\pi_f(2)\ne \pi_g(2)$
\end{enumerate}
Then
\[
\pi_f(m-1)\ne\pi_g(m-1)\bmod m\ \;\text{ and }\;\; \pi_g(2) = \pi_f(m)- \pi_f(m-1)\bmod m.
\]
\end{lemma}
\begin{proof}
    Using the Lemma \ref{lem:secondcomponent}, we can see that 
    \[\pi_g(m-1) = \pi_f(m)- \pi_f(2)\bmod m,\qquad \pi_f(m-1) = \pi_g(m)- \pi_g(2)\bmod m.\]
    From 2nd hypothesis $\pi_f(m)- \pi_f(2)\not = \pi_g(m)- \pi_g(2)\bmod m$ and therefore, $\pi_g(m-1)\not =\pi_f(m-1)\bmod m$.
    Moreover, by using the above equivalency  we can write \[\pi_g(2) =\pi_g(m) - \pi_f(m-1) = \pi_f(m) - \pi_f(m-1)\bmod m.\]
\end{proof}
\begin{lemma}\label{lem:2_m-1}
Let $f,g$ be lag-homometric permutations of $\{1,\dots,m\}$.  Assume
\begin{enumerate}
  \item $\pi_f(1)=\pi_g(1)=0$,
  \item $\pi_f(m)=\pi_g(m)$,
  \item $\pi_f(2)=\pi_g(2)$.
\end{enumerate}
Then $\pi_f(m-1)=\pi_g(m-1)$.
\end{lemma}

\begin{proof}
Assume, for contradiction, that $\pi_f(2)=\pi_g(2)$ but $\pi_f(m-1)\neq\pi_g(m-1)$.
Set
\[
d := m-2 .
\]
Observe that among ordered label-pairs from $\{1,\dots,m\}$ the only pairs whose label-difference equals $d$ are
\[
(1,m-1)\quad\text{and}\quad(2,m).
\]
Hence the value $d$ can occur in the shift-multisets only at shifts coming from these two pairs.

Write
\[
k_1^f := \pi_f(m-1),\qquad k_2^f := \pi_f(m)-\pi_f(2) \bmod m,
\]
and likewise
\[
k_1^g := \pi_g(m-1),\qquad k_2^g := \pi_g(m)-\pi_g(2)\bmod m.
\]
By construction, in permutation $f$ the difference $d$ appears (as an element of the multiset) exactly at the shifts in the multiset $\{k_1^f,k_2^f\}$ (counted with multiplicity), and in $g$ it appears exactly at the shifts in $\{k_1^g,k_2^g\}$.

Hypotheses (2) and (3) give $\pi_f(m)=\pi_g(m)\bmod m$ and $\pi_f(2)=\pi_g(2)\bmod m$, hence
\[
k_2^f = \pi_f(m)-\pi_f(2)=\pi_g(m)-\pi_g(2)=k_2^g\bmod m.
\]
Being lag-homometric  implies that the multiset of shifts where $d$ occurs in $f$ equals that for $g$, i.e.
\(
\{k_1^f,k_2^f\} = \{k_1^g,k_2^g\}\quad\text{(as multisets).}
\)
Since $k_2^f=k_2^g\bmod m$, two cases cover all possibilities:
\begin{itemize}
    \item \textbf{Case 1:}  $k_1^f\neq k_2^f$: \\
    Then the multiset equality forces $k_1^f=k_1^g$, i.e.\ $\pi_f(m-1)=\pi_g(m-1)\bmod m$, contradicting the assumption.
\item \textbf{Case 2:} $k_1^f=k_2^f$.\\
Then
\(
\pi_f(m-1)=\pi_f(m)-\pi_f(2)\bmod m.
\)
By multiset equality we must also have $k_1^g=k_2^g\bmod m$, hence
\(
\pi_g(m-1)=\pi_g(m)-\pi_g(2)\bmod m.
\)
Using $\pi_f(m)=\pi_g(m)\bmod m$ and $\pi_f(2)=\pi_g(2)\bmod m$ we obtain
\(
\pi_g(m-1)=\pi_f(m)-\pi_f(2)=\pi_f(m-1)\bmod m,
\)
again contradicting the assumption.
\end{itemize}
In both cases we see a contradiction. Therefore the assumption $\pi_f(m-1)\neq\pi_g(m-1)\bmod m$ is false, and $\pi_f(m-1)=\pi_g(m-1)\bmod m$.
\end{proof}

\begin{lemma}\label{lem:helper}
Let $f,g$ be lag-homometric permutations of $\{1,\dots,m\}$.  Assume
\begin{enumerate}
  \item $\pi_f(1)=\pi_g(1)=0$,
  \item $\pi_f(m)=\pi_g(m)$.
\end{enumerate}
Then for all $p\in\{2,\ldots,m\}$:
\[\forall a\in \{1,2,\ldots, p\},\quad \pi_f(a) = \pi_g(a) \Longrightarrow \pi_f(m+1-p) = \pi_g(m+1-p)\]
\end{lemma}
\begin{proof}
    We will prove this by induction.\\
 \noindent \textbf{Base case} $p = 2$: Done in the Lemma \ref{lem:2_m-1}.\\
 \noindent \textbf{Induction hypothesis (IH)}:
For any $1\leq a\leq p-1$, the equation $\pi_f(a) = \pi_g(a)$ implies that $\pi_f(m+1-a) = \pi_g(m+1-a)$.\\
 \noindent \textbf{Induction step}:
    We prove that if $\pi_f(p) = \pi_g(p)$ then $ \pi_f(m+1-p) = \pi_g(m+1-p)$.
    We set $d = m- p$. The only pairs that can generate the $d$ are :
    \[(1,1+d),\ (2, 2+d),\ \ldots,\ (p, p+d).\]
    Now, we define $k_a^f = \pi_f(a+d) - \pi_f(a),\  k_a^g = \pi_g(a+d) - \pi_g(a)$. Then we set 
    \[M_f=\{k_a^f\ |\ a\in\{1,2,\ldots, p\} \},\qquad M_g=\{k_a^g\ |\ a\in\{1,2,\ldots, p\}\}\]
    Since we have $\Delta_f = \Delta_g$, we can claim $M_f = M_g$. This is true because the only possible shifts that generate the $d$ are the shift values in these sets. By the hypothesis of the induction step, we know that $\pi_f(p) = \pi_g(p)$ and by the third hypothesis of the statement of the lemma we know that  $\pi_f(m) = \pi_g(m)$. Hence $k_p^f = k_p^g$. Now, we define $\alpha = a+d$ and $\beta = m+1-\alpha$. Then we have
    \[m+1-\beta = m+1 - (a + (m-p)) = p-a+1\]
    and because $2\leq a\leq p-1$ we get 
    \[2=p-(p-1)+1\leq p-a+1\leq p-2+1 = p-1\]
    Therefore,
    \[\pi_f(m+1-\beta) = \pi_g(m+1-\beta)\]
    which implies that $\pi_f(a+d) = \pi_g(a+d)$. Hence,
    \[k_a^f = \pi_f(a+d)-\pi_f(a) = \pi_g(a+d) - \pi_g(a) = k_a^g.\]
    Since we show that $k_2^f = k_2^g, k_3^f = k_3^g,\ldots, k_p^f = k_p^g$ and $M_f = M_g$, therefore we can say that $k_1^f$ has to be equal $k_1^g$ which implies that 
    \[\pi_f(1+d)-\pi_f(1) = \pi_g(1+d) - \pi_g(1),\]
    and from second hypothesis we get the $\pi_f(1+d)= \pi_g(1+d)$ and that is 
    $\pi_f(1+m-p)= \pi_g(1+m-p)$. Therefore, the induction step is proven and the result is as desired. 
\end{proof}

\begin{theorem}\label{thm:finalPermutation}
Let $f,g$ be permutations of $\{1,\dots,m\}$.
Assume
\begin{enumerate}
  \item $\pi_f(1)=\pi_g(1)=0$,
  \item $\pi_f(m)=\pi_g(m)$,
  \item $\Delta_f(k)=\Delta_g(k)$ for every shift $k\in\mathbb Z_m$.
\end{enumerate}
Then for every $p\in\{2,\ldots, m-1\}$
$$\pi_f(p)\ne \pi_g(p)\Longrightarrow\pi_g(m+1-p) = \pi_f(m)- \pi_f(p) \pmod m.$$
\end{theorem}

\begin{proof}
We prove this by induction on $p$.

\noindent \textbf{Base case} ($p=2$): This is Lemma~\ref{lem:secondcomponent}.

\noindent \textbf{Induction hypothesis}: For every $1\leq a\leq p-1$,
\[
\pi_f(a)\ne \pi_g(a) \;\Longrightarrow\; \pi_g(m+1-a) = \pi_f(m) - \pi_f(a).
\]

\noindent \textbf{Induction step}: Assume $\pi_f(p)\ne \pi_g(p)$. We must show $\pi_g(m+1-p) = \pi_f(m) -\pi_f(p)$.

Set $d = m-p$. For each $a \in \{1,\dots,p\}$, define
\[
k_a^f = \pi_f(a+d) - \pi_f(a), \qquad k_a^g = \pi_g(a+d) - \pi_g(a),
\]
and collect these into multisets
\[
M_f = \{k_1^f, k_2^f, \ldots, k_p^f\}, \qquad M_g = \{k_1^g, k_2^g, \ldots, k_p^g\}.
\]
Since $\Delta_f(k) = \Delta_g(k)$ for all $k$, we have $M_f = M_g$ as multisets.

Since $\pi_f(p)\ne \pi_g(p)$ and $\pi_f(m) = \pi_g(m)$, we compute
\begin{align*}
k_p^f = \pi_f(p+d) - \pi_f(p) &= \pi_f(m) - \pi_f(p) \\
&\ne \pi_f(m) -\pi_g(p) = \pi_g(m) -\pi_g(p) = \pi_g(p+d) - \pi_g(p) = k_p^g.
\end{align*}
Therefore $k_p^f \ne k_p^g$.

We now show that every $k_a^g$ for $a \in \{2,\dots,p-1\}$ is determined by the induction hypothesis, leaving $k_1^g$ as the only free entry in $M_g$.

For each $a \in \{2,\dots,p-1\}$, define $\beta = p+1-a$. Since $2 \le a \le p-1$, we also have $2 \le \beta \le p-1$. A direct computation gives
\[
a + d = m + 1 - \beta, \qquad \beta + d = m + 1 - a.
\]
This means $k_a^g$ involves $\pi_g(a+d) = \pi_g(m+1-\beta)$, and $k_\beta^g$ involves $\pi_g(\beta+d) = \pi_g(m+1-a)$. These are exactly the quantities the IH provides formulas for when positions disagree. Since both $a < p$ and $\beta < p$, the IH applies to each. We consider four cases depending on whether $f$ and $g$ agree or disagree at $a$ and at $\beta$.

\medskip
\noindent (i): $\pi_f(a) = \pi_g(a)$ and $\pi_f(\beta) = \pi_g(\beta)$.
Both positions agree. From $\pi_f(\beta) = \pi_g(\beta)$, Lemma~\ref{lem:helper} gives $\pi_f(m+1-\beta) = \pi_g(m+1-\beta)$, i.e., $\pi_f(a+d) = \pi_g(a+d)$. Combined with $\pi_f(a) = \pi_g(a)$:
\[
k_a^g = \pi_g(a+d) - \pi_g(a) = \pi_f(a+d) - \pi_f(a) = k_a^f.
\]
So $k_a^g$ is determined. By the symmetric argument, $k_\beta^g = k_\beta^f$ is also determined.

\medskip
\noindent (ii): $\pi_f(a) = \pi_g(a)$ and $\pi_f(\beta) \ne \pi_g(\beta)$.

Since $\pi_f(\beta) \ne \pi_g(\beta)$ and $\beta < p$, the IH gives
\[
\pi_g(m+1-\beta) = \pi_f(m) - \pi_f(\beta).
\]
Since $a + d = m + 1 - \beta$, this says $\pi_g(a+d) = \pi_f(m) - \pi_f(\beta)$. We also know $\pi_g(a) = \pi_f(a)$. Therefore
\[
k_a^g = \pi_g(a+d) - \pi_g(a) = \bigl(\pi_f(m) - \pi_f(\beta)\bigr) - \pi_f(a).
\]
Everything on the right is known. So $k_a^g$ is determined.

\medskip
\noindent (iii): $\pi_f(a) \ne \pi_g(a)$ and $\pi_f(\beta) = \pi_g(\beta)$.

Since $\pi_f(a) \ne \pi_g(a)$ and $a < p$, the IH gives
\[
\pi_g(m+1-a) = \pi_f(m) - \pi_f(a).
\]
Since $\beta + d = m+1-a$, this says $\pi_g(\beta+d) = \pi_f(m) - \pi_f(a)$. We also know $\pi_g(\beta) = \pi_f(\beta)$. Therefore
\[
k_\beta^g = \pi_g(\beta+d) - \pi_g(\beta) = \bigl(\pi_f(m) - \pi_f(a)\bigr) - \pi_f(\beta).
\]
So $k_\beta^g$ is determined. Moreover, from $\pi_f(\beta) = \pi_g(\beta)$, Lemma~\ref{lem:helper} gives $\pi_f(a+d) = \pi_g(a+d)$, so $k_a^g$ is determined as well.

\medskip
\noindent (iv): $\pi_f(a) \ne \pi_g(a)$ and $\pi_f(\beta) \ne \pi_g(\beta)$.

Both disagree. The IH applies to both:
\[
\pi_g(m+1-\beta) = \pi_f(m) - \pi_f(\beta), \qquad \pi_g(m+1-a) = \pi_f(m) - \pi_f(a).
\]
Since $a+d = m+1-\beta$ and $\beta+d = m+1-a$, we get
\[
k_a^g = \bigl(\pi_f(m) - \pi_f(\beta)\bigr) - \pi_g(a), \qquad
k_\beta^g = \bigl(\pi_f(m) - \pi_f(a)\bigr) - \pi_g(\beta).
\]
Both are determined.

\medskip
In all four cases, $k_a^g$ is determined for every $a \in \{2,\dots,p-1\}$. The entry $k_p^g = \pi_g(m) - \pi_g(p)$ is also known. The only entry of $M_g$ that is not yet determined is
\[
k_1^g = \pi_g(1+d) - \pi_g(1) = \pi_g(m+1-p),
\]
since $\pi_g(1) = 0$.

Now, $M_f = M_g$ as multisets. The value $k_p^f$ appears in $M_f$, but $k_p^g \ne k_p^f$ and all of $k_2^g, \dots, k_{p-1}^g$ are already fixed. The only entry that can match $k_p^f$ is $k_1^g$. Therefore $k_1^g = k_p^f$, which gives
\[
\pi_g(m+1-p) = k_1^g = k_p^f = \pi_f(m) - \pi_f(p). \qedhere
\]
\end{proof}

\begin{theorem}
    Fix $m \in \mathbb{Z}_{\ge 1}$.
    Let $P$ be the set of all permutations of any consecutive $m$ numbers, and let 
    $V_f^{P}$ denote the restriction of $V_f$ to $P$.  
    Then 
    \[
        g = \Rot^r_x(f) \text{ or } g = \Rot^r_x(\RoC(f)) \quad\text{ for some } x\in\cms\Longleftrightarrow V_f^{P} = V_g^{P}
    \]
\end{theorem}

\begin{proof}
    The forward direction is immediate:  
    the discrete signature is invariant under shifts, rotations, and
    Reverse-of-Complement, so any such transformation leaves $V_f^{P}$ unchanged.

    For the converse direction, assume $f, g \in P$ and $V_f^{P} = V_g^{P}$.
    By Lemma~\ref{lem:expo_multisets}, equality of the restricted signatures
    forces $f$ and $g$ to be \emph{lag-homometric}.

    Since $f, g \in P$, each is obtained from some $f', g' \in P'$ 
    by adding a constant and applying a rotation.  
    Both operations preserve lag-homometricity, so the corresponding elements
    $f', g' \in P'$ are also lag-homometric.

    Now normalize each $f', g'$ to a canonical representative in $P'$:
    let
    \[
        i_{\min} := \arg\min f', \qquad
        f'' := \Rot^r_{\,m - i_{\min}}\!\bigl(f' - \min f' + 1\bigr),
    \]
    so that $f''(0)=1$ and $f''$ becomes a permutation of $\{1,\ldots,m\}$.  
    Apply the same normalization to $g'$ to obtain $g''$.
    Again, all steps preserve lag-homometricity, hence
    \[
        \Delta_{f''}(k) = \Delta_{g''}(k)
        \qquad\text{for all } k.
    \]

    At this point, $f''$ and $g''$ are permutations in $P'$ that are 
    lag-homometric.  
    Therefore, by Theorem~\ref{thm:finalPermutation},
    the signature on $P'$ is injective up to rotation or rotation of 
    the Reverse-of-Complement, which yields
    \[
        g'' = \Rot^r_x(f'')
        \quad\text{or}\quad
        g'' = \Rot^r_x(\RoC(f''))
        \qquad\text{for some } x \in \cms.
    \]

    Undoing the normalization and reversing the earlier shifts and rotations
    shows that the same relation holds for $f$ and $g$:
    \[
        g = \Rot^r_x(f)
        \quad\text{or}\quad
        g = \Rot^r_x(\RoC(f)),
    \]
    for some $x \in \cms$, completing the proof.
\end{proof}


\section{Invariance of General Position Functions}
\label{app:unique-inv}

This section provides details on the proof of invariance of lag-homometric functions $f,g \in \F^m_{\neq}$ that are in \emph{general position}, in the sense that all values $f(j) \neq f(j')$ for $j \neq j' \in [m]$, and also all differences $f(j)-f(j') \in \Delta_f$ for $j\neq j'$ are also distinct.  

For $f \in \F^m_{\neq}$ let $\Delta_f = \bigcup_{k \in [m]} \Delta_f(k)$ be the set of all differences.  The difference of a function value $f(j)$ with itself $f(j) - f(j) = 0$, but by definition of $\F^m_{\neq}$, all other differences are distinct.  Thus $\Delta_f$ has $m(m-1) + 1$ distinct values (with the $1$ coming from the $m$ copies of difference of $0$).  

Recall that we define the \emph{hat} version $\hat f$ of a function $f \in \F^{m}_{\neq}$ as being in order by its index, so $0 = \hat f(0) < \hat f(1) < \hat f(2) < \ldots < \hat f(m-2) < \hat f(m-1) = M$.  

\begin{lemma}\label{lem:2secondlargest}
Let $f=[a_0=0,a_1,\ldots,a_{m-1}]$ be a function on $\mathbb{Z}_m$ and let 
$\hat f(0)=0<\hat f(1)<\cdots<\hat f(m-1)=M$ be the increasing rearrangement of 
$\{a_0,\ldots,a_{m-1}\}$. For $k\in\{0,1,\ldots,m-1\}$ define
\[
\Delta_f(k)=\{\,f(j+k)-f(j)\mid j=0,1,\ldots,m-1\,\}.
\]
Then the second largest value in $\Delta_f$ can only be one of
\[
\hat f(m-1)-\hat f(1), 
\qquad 
\hat f(m-2)-\hat f(0).
\]
\end{lemma}

\begin{proof}
Each element of $\Delta_f(k)$ is a difference of the form 
$\hat f(j)-\hat f(i)$ for some $0\le i,j\le m-1$.  
Any difference using only indices $0<i,j<m-2$ is strictly smaller than every 
difference involving either the minimum $\hat f(0)$ or the maximum $\hat f(m-1)$.  
Thus the second largest value cannot arise from such pairs.

Hence the second largest value must come from one of the two families
\[
\hat f(m-1)-\hat f(i), \quad i=1,2,\ldots,m-2,
\qquad\text{or}\qquad
\hat f(i)-\hat f(0), \quad i=1,2,\ldots,m-2.
\]

Within the first family,
\[
\hat f(m-1)-\hat f(1) \;>\; \hat f(m-1)-\hat f(i),
\qquad i=2,\ldots,m-2,
\]
so the largest candidate there is $\hat f(m-1)-\hat f(1)$.

Within the second family,
\[
\hat f(m-2)-\hat f(0) \;>\; \hat f(i)-\hat f(0),
\qquad i=1,2,\ldots,m-3,
\]
so the largest candidate there is $\hat f(m-2)-\hat f(0)$.

Thus the second largest value in $\bigcup_k \Delta_f(k)$ can only be 
$\hat f(m-1)-\hat f(1)$ or $\hat f(m-2)-\hat f(0)$.
\end{proof}

\begin{lemma}\label{lem:secondlarge}
Let $f,g \in \F^m_{\neq}$ be lag-homometric, and in canonical form so  
$\hat f(0)=\hat g(0)=0$.  
Then we must have
\[
\hat g(m-2)=\hat f(m-2)
\quad\text{XOR} \quad
\Big(\hat g(m-2)=\hat f(m-1)-\hat f(1)\quad \&\quad \hat f(m-2)=\hat g(m-1)-\hat g(1)\Big).
\]
\end{lemma}

\begin{proof}
Let $M$ be the unique maximum difference in $\Delta_f$. Because $f,g$ are lag-homometric, they share the exact same multiset of differences for every lag $k$. Thus, $f$ and $g$ share the same overall maximum difference $M$, and the same second largest difference $d_2$. 

By Lemma~\ref{lem:2secondlargest}, the second largest difference in $f$ must be either $A_f = \hat f(m-1)-\hat f(1)$ or $B_f = \hat f(m-2)-\hat f(0) = \hat f(m-2)$. Similarly, the second largest difference in $g$ must be either $A_g = \hat g(m-1)-\hat g(1)$ or $B_g = \hat g(m-2)$. Because $d_2$ is uniquely the second largest difference in both functions, $d_2 \in \{A_f, B_f\}$ and $d_2 \in \{A_g, B_g\}$.

Because $f \in \F^m_{\neq}$, all non-zero differences are uniquely generated by exactly one pair of indices. Therefore, the value $d_2$ occurs at a unique lag $k_2$ in $f$. Since $\Delta_f(k_2) = \Delta_g(k_2)$, the exact value $d_2$ must also be generated in $g$ at the exact same lag $k_2$.
The lag for $A_f$ is $k_{A_f} = \pi_f(\hat f(m-1)) - \pi_f(\hat f(1)) \pmod m$. 
The lag for $B_f$ is $k_{B_f} = \pi_f(\hat f(m-2)) - \pi_f(\hat f(0)) \pmod m$.
(And similarly for $A_g$ and $B_g$).

We have two possibilities for $f$:
\begin{enumerate}
    \item $d_2 = A_f$. Then $d_2$ occurs at lag $k_{A_f}$. In $g$, $d_2$ must be either $A_g$ or $B_g$. If $d_2 = A_g$, then $A_f = A_g$ (which implies $\hat f(1) = \hat g(1)$ since $M$ is shared), and their lags must match: $k_{A_f} = k_{A_g}$. If $d_2 = B_g$, then $A_f = B_g$, and their lags must match: $k_{A_f} = k_{B_g}$.
    \item $d_2 = B_f$. Then $d_2$ occurs at lag $k_{B_f}$. In $g$, $d_2$ is either $A_g$ or $B_g$. If $d_2 = B_g$, then $B_f = B_g$, meaning $\hat f(m-2) = \hat g(m-2)$. If $d_2 = A_g$, then $B_f = A_g$.
\end{enumerate}

Therefore, we definitively establish that either $\hat f(m-2) = \hat g(m-2)$ (if $B_f = B_g$) OR $\hat g(m-2) = \hat f(m-1)-\hat f(1)$ and $\hat f(m-2) = \hat g(m-1)-\hat g(1)$ (the crossed case where $B_g = A_f$ and $B_f = A_g$). These cases are mutually exclusive (XOR) because if both were true, then $\hat f(m-2) = \hat f(m-1)-\hat f(1)$, which violates the distinctness condition of $\F^m_{\neq}$. 
\end{proof}

Recall, that since we have identified the location of the second largest difference in $\Delta_f = \Delta_g$, we can convert $f,g$ to $\RoC$ canonical form.  This ensures that the second largest difference is $\hat f(m-2) - \hat f(0) = \hat f(m-2)$; if it is not originally, we apply $\RoC$, and then shift and rotate to make in canonical form.  Moreover, since we know for which $k$ this second largest value is in $\Delta_f(k) = \Delta_g(k)$, we can identify its index $j$ in $f$ (so $f(j) = \hat f(m-2)$ and it is the same for $g$).    

In general, define $\pi_f : \{f(j)\}_{j \in [m]} \to [m]$ which maps from the function values of $f$ (which since $f \in \F^m_{\neq}$ are unique) back to the index of $f$.  That is $\pi_f(f(j)) =j$.  This notation will be useful when applied to $\hat f$; so for instance $f(\pi_f(\hat f(m-2)) = \hat f(m-2)$.  So far, by having $f,g$ in $\RoC$-canonical form and lag-homometric, we have identified the location of $\pi_f(\hat f(0)) = \pi_g(\hat g(0))$ (which is index $0$), $\pi_f(\hat f(m-1)) = \pi_g(\hat g(m-1))$ (the largest value) and $\pi_f(\hat f(m-2)) = \pi_g(\hat g(m-2))$ (the second largest value).

Next, for two positive integer $s,t$, we define $\Delta_f^{s,t}$ which will be a subset of the distinct distances in $\Delta_f$.  This assumes we have identified the location of the $t$ largest and $s$ smallest values of $f$, and we want to exclude all resolved differences between 
\begin{itemize}
\item  $f(j) - f(j')$ for $j,j' < s$, 
\item  $f(m-j) - f(j')$ for $j < t$ and $j'< s$, 
\item  $f(j)-f(m-j')$ for $j < s$ and $j'< t$, or
\item  $f(m-j) - f(m-j')$ for $j,j' < t$;  
\end{itemize}
so then the remaining differences are in $\Delta_f^{s,t}$.  
After Lemma \ref{lem:secondlarge} we have resolved the $s=1$ smallest, and $t=2$ largest and so can consider $\Delta_f^{1,2} = \Delta_g^{1,2}$.

\begin{lemma} \label{lem:Dst-induction-app}
    Consider $f,g \in \F^m_{\neq}$ which are lag-homometric and in $\RoC$-canonical form.  Assume we have resolved the locations of the $s \geq 1$ smallest and $t \geq 2$ largest values in $f,g$, confirming they are the same.  We can then confirm the location of either the $(s+1)$th smallest or $(t+1)$th largest by value of $f,g$ and confirm that this value matches.  
\end{lemma}
\begin{proof}
    Let $\beta = \max \Delta_f^{s,t} = \max \Delta_g^{s,t}$.  This difference must either be formed by 
    \[
    \beta = \hat f(m-1) - \hat f(s)  \quad  \text{ or } \quad \beta = \hat f(m-1-t) - \hat f(0) = \hat f(m-1-t)
    \]
    in $f$, and similar in $g$ must either be
    \[
    \beta = \hat g(m-1) - \hat g(s)  \quad  \text{ or } \quad \beta = \hat g(m-1-t) - \hat g(0) = \hat g(m-1-t).
    \]
    If in both $f$ and $g$ these are the first case, then since we know the location of $f(m-1)$ and which $\Delta_f(k) = \Delta_g(k)$ contains $\beta$ we can resolve the location of $g(s)$, resolving the location of the $(s+1)$th smallest.  
    Similarly, if both $f$ and $g$ are in the second case, then we can resolve the location of the $(t+1)$th largest.  

    What remains is to show that if they have a mismatch ($\beta = \hat f(m-1) - \hat f(s) = \hat g(m-1-t)$, or $\beta = \hat f(m-1-t) = \hat g(m-1) - \hat g(s) $), this causes a contradiction, so that case cannot happen.  These cases are symmetric, so we analyze the condition $\beta = \hat f(m-1) - \hat f(s) = \hat g(m-1-t)$.  

    Define the position function $p_f(j) = \pi_f(\hat f(j))$ and $p_g(j) = \pi_g(\hat g(j))$.  
    Now we must have $g(\pi_g(\beta)) \neq f(\pi_g(\beta))$, since otherwise $f(\pi_g(\beta)) - f(0) = \beta$ and $f(0) - f(\pi_g(\beta)) = -\beta$ would both be in $\Delta_f$, but $\beta = \hat f(m-1) - \hat f(s)$ is already accounted for through a different pair of indices, violating the uniqueness of differences in $\F^m_{\neq}$.  
    
    Now consider the lag $k = p_g(m-2) - \pi_g(\beta) \pmod m$.  Since we established $p_f(m-2) = p_g(m-2)$, both $f$ and $g$ generate a difference at this lag involving position $p_f(m-2) = p_g(m-2)$ and position $\pi_g(\beta)$.  Specifically:
    \[
    \delta_1 = f(p_f(m-2)) - f(\pi_g(\beta)) = \hat f(m-2) - f(\pi_g(\beta)) \in \Delta_f(k),
    \]
    \[
    \delta_2 = g(p_g(m-2)) - g(\pi_g(\beta)) = \hat g(m-2) - g(\pi_g(\beta)) \in \Delta_g(k).
    \]
    Since $\hat f(m-2) = \hat g(m-2)$ but $f(\pi_g(\beta)) \neq g(\pi_g(\beta))$, we have $\delta_1 \neq \delta_2$.  
    
    Now, lag-homometricity requires $\Delta_f(k) = \Delta_g(k)$ as multisets.  In particular, $\delta_2 \in \Delta_g(k)$ must also appear in $\Delta_f(k)$.  But in $\F^m_{\neq}$, each difference value at a given lag is produced by exactly one pair of indices.  The only pair in $\Delta_f(k)$ involving position $p_f(m-2)$ produces $\delta_1 \neq \delta_2$.  So $\delta_2$ would need to be produced by a \emph{different} pair $(j,j')$ with $j - j' \equiv k \pmod m$.  However, uniqueness of all nonzero differences in $\F^m_{\neq}$ means $\delta_2$ can only occur once across all of $\Delta_f$, and if it occurred at lag $k$ it would need a specific pair, but the only pair at lag $k$ subtracting from a value equal to $\hat f(m-2)$ is the one already used.  This is a contradiction, establishing that the mismatch case cannot occur.
\end{proof}

Now we can use Lemma \ref{lem:Dst-induction-app} to complete the argument.  As mentioned, $\RoC$-canonical gives us $\Delta_f^{s,t}$ with $s=1, t=2$.  This satisfies the initial conditions of Lemma \ref{lem:Dst-induction-app}, and when we invoke that lemma, it increases either $s$ or $t$.  This continues until $s+t = m$, in which case all positions of $f$ and $g$ are confirmed to be the same, completing the proof.   We can then state the following theorem.  

\begin{theorem}\label{thm:fg-same-app}
  For $f,g \in F^m_{\neq}$ which are lag-homometric in $\RoC$-canonical form, then they must be equal.      
\end{theorem}

\section{Sampling from RKHS}
\label{app:RKHS}

\subsection{Kernels and Hilbert Space}
\label{app:kernels}
In this section, we define our representation and highlight some of its key properties. Specifically, the representation maps every bounded function into a reproducing kernel Hilbert space (RKHS), where the induced positive definite kernel allows us to compare functions in $L^2(S^1)$ in a principled and geometrically meaningful way.

We now define an inner product on the space of signatures, which is the standard inner product of a Hilbert space.

\begin{definition}[Invariant Signature Kernel]
The \emph{Invariant Signature Kernel} \(k(f,g)\) is defined as the inner product between two signature functions \(V_f\) and \(V_g\) in \(L^2(S^1)\):
\[
k(f,g) = \langle V_f, V_g \rangle_{L^2} = \int_0^{2\pi} V_f(\alpha) \overline{V_g(\alpha)} \, d\alpha.
\]
This kernel quantifies the total similarity between the functions \(f\) and \(g\) across all possible shifts on \(S^1\).
\end{definition}

The invariant signature kernel \(k(f,g) = \langle V_f, V_g \rangle_{L^2}\) provides a powerful tool for quantifying the similarity between functions in \(L^2(S^1)\). By measuring the inner product of signature functions \(V_f\) and \(V_g\), the kernel captures the degree of alignment between \(f\) and \(g\) across all possible shifts on the unit circle \(S^1\). For a kernel to effectively serve as a similarity measure, it must be \emph{positive semi-definite}. Positive definiteness ensures that the kernel induces a valid inner product in a reproducing kernel Hilbert space (RKHS), which has several critical implications. First, it guarantees that the kernel defines a notion of distance via the induced norm, allowing meaningful comparisons of similarity: larger values of \(k(f,g)\) indicate greater similarity, while \(k(f,f) = \|V_f\|_{L^2}^2 \geq 0\) reflects the magnitude of \(f\). Positive semidefiniteness ensures the distance is a pseudometric.  
Second, positive definiteness ensures that the kernel matrix formed by evaluating \(k(f_i, f_j)\) for a set of functions is positive semi-definite, which is essential for applications in machine learning, such as kernel methods, where this property ensures numerical stability and well-defined optimization problems. Without positive semi-definiteness, the kernel could produce negative or undefined similarity measures, undermining its interpretability and utility. The following theorem establishes that our invariant signature kernel satisfies this crucial property.

\begin{remark}[Kernel as an Expectation]
Our kernel, \(k(f,g)\), can be expressed as an expected value. Let \(\alpha\) be a random variable drawn from the uniform distribution on the interval \([0, 2\pi)\). The probability density function (PDF) for this distribution is \(p(\alpha) = \frac{1}{2\pi}\) for any \(\alpha \in [0, 2\pi)\).

By the definition of expected value, we have:
\[
\mathbb{E}_{\alpha \sim U[0, 2\pi]} \left[ V_f(\alpha) \cdot V_g(\alpha) \right] = \int_0^{2\pi} \left( V_f(\alpha) \cdot V_g(\alpha) \right) \cdot p(\alpha) \,d\alpha = \frac{1}{2\pi} \int_0^{2\pi} V_f(\alpha) \cdot V_g(\alpha) \,d\alpha.
\]
Rearranging this equation gives a direct relationship between our kernel and the expectation:
\[
k(f,g) = \int_0^{2\pi} V_f(\alpha) \cdot V_g(\alpha) \,d\alpha = 2\pi \cdot \mathbb{E}_{\alpha \sim U[0, 2\pi]} \left[ V_f(\alpha) \cdot V_g(\alpha) \right].
\]
This formulation allows us to use Monte Carlo methods, such as the one proposed by \cite{RahimiRecht2007}, to approximate the kernel.
\end{remark}

\begin{remark}[Approximation Bound via Hoeffding's Inequality]
We can determine the number of random samples \(D\) required to ensure that our randomized kernel \(Z_f^T Z_g\) provides a good approximation to the true kernel \(k(f,g)\) by applying Hoeffding's inequality (a form of the Chernoff bound for bounded random variables).

First, we define a practical, finite-dimensional representation. The \emph{randomized signature} \(Z_f \in \mathbb{R}^D\) is constructed by sampling \(D\) i.i.d. shifts \(\{\alpha_1, \dots, \alpha_D\}\) uniformly from \([0, 2\pi]\) and setting:
\[
Z_f = \sqrt{\frac{2\pi}{D}} \begin{bmatrix} V_f(\alpha_1) & V_f(\alpha_2) & \cdots & V_f(\alpha_D) \end{bmatrix}^T.
\]
The inner product is then:
\[
Z_f^T Z_g = \frac{2\pi}{D} \sum_{i=1}^D V_f(\alpha_i) V_g(\alpha_i),
\]
which serves as an unbiased Monte Carlo estimator of the true kernel \(k(f,g)\). To bound the approximation error, we apply Hoeffding's inequality with the following identifications:
\begin{itemize}
    \item The number of samples is \(n = D\).
    \item The i.i.d. random variables are \(X_i = 2\pi \, V_f(\alpha_i) \, V_g(\alpha_i)\).
    \item The sample mean \(\bar{X} = \frac{1}{D} \sum_{i=1}^D X_i\) equals the randomized kernel \(Z_f^T Z_g\).
    \item The expected value \(\mathbb{E}[\bar{X}]\) is the true kernel \(k(f,g)\).
\end{itemize}
Since \(f\) and \(g\) are bounded, implying that the signature functions \(V_f(\alpha)\) and \(V_g(\alpha)\) are also bounded. Consequently, each \(X_i\) is bounded within an interval of length \(\Delta\) for some \(\Delta > 0\).

Substituting into Hoeffding's inequality yields:
\[
\Pr\left[ |Z_f^T Z_g - k(f,g)| > \varepsilon \right] \le 2 \exp\left( -\frac{2 \varepsilon^2 D}{\Delta^2} \right).
\]
To achieve a confidence level of \(1 - \delta\), we set the right-hand side equal to \(\delta\) and solve for \(D\):
\[
D \ge \frac{\Delta^2}{2 \varepsilon^2} \ln\left( \frac{2}{\delta} \right).
\]
This demonstrates that the required number of samples scales as \(O\left( \frac{\Delta^2}{\varepsilon^2} \ln \frac{1}{\delta} \right)\).

Moreover, if we normalize the functions such that \(\Phi(f(\theta) - f(\theta + \alpha)) < \frac{1}{2\pi \sqrt{2\pi}}\) for all \(\alpha, \theta \in [0, 2\pi)\), then:
\[
\int_0^{2\pi} \Phi(f(\theta) - f(\theta + \alpha)) \, d\theta \le \int_0^{2\pi} \frac{1}{2\pi \sqrt{2\pi}} \, d\theta = \frac{1}{\sqrt{2\pi}},
\]
which implies \(V_f(\alpha) \le \frac{1}{\sqrt{2\pi}}\) and similarly for \(V_g(\alpha)\). Thus, \(X_i \le 1\), allowing us to set \(\Delta = 1\) (assuming \(X_i \ge 0\)). For \(\varepsilon = 0.1\) and \(\delta = 0.01\), this yields \(D \ge 265\) (approximately).
\end{remark}

\begin{remark}[Bounding $X_i$ for $\Phi(x)=e^{-\lambda x}$]
\label{rem:bound_lambda}
Fix $\lambda>0$. For
\[
V_f(\alpha)=\int_0^{2\pi} e^{-\lambda\,(f(\theta)-f(\theta+\alpha))}\,d\theta,
\qquad
X_i=2\pi\,V_f(\alpha_i)\,V_g(\alpha_i),
\]
we require a uniform bound $|X_i|\le B$.  

A direct estimate uses $\;|f(\theta)-f(\theta+\alpha)|\le 2\|f\|_\infty\;$, giving
\[
\|V_f\|_\infty \;\le\;\int_0^{2\pi} e^{2\lambda \|f\|_\infty}\,d\theta
=2\pi\,e^{2\lambda\|f\|_\infty},
\]
and hence
\[
|X_i|\;\le\;B:=2\pi\|V_f\|_\infty\|V_g\|_\infty
\;\le\;8\pi^3\,e^{2\lambda(\|f\|_\infty+\|g\|_\infty)}.
\]
\end{remark}


\section{Error Bounds and Sampling Requirements for the FFT-Based Computation}

\begin{remark}[Error Bound of the Discretized FFT Evaluation]
\label{rem:fft_error_bound}
The continuous signature \(V_f(\alpha) = \int_0^{2\pi} e^{-\lambda(f(\theta) - f(\theta+\alpha))}\,d\theta\) must be numerically approximated by a discrete periodic convolution computed via the FFT. Let
\[
V_f^{(N_\theta)}(\alpha_k)
= \Delta\theta \sum_{j=0}^{N_\theta-1} e^{-\lambda(f(\theta_j) - f(\theta_{j+k}))},
\qquad
\theta_j = \tfrac{2\pi j}{N_\theta},
\quad
\Delta\theta = \tfrac{2\pi}{N_\theta}.
\]
Then \(V_f^{(N_\theta)}\) converges to \(V_f\) as \(N_\theta \to \infty\). To quantify this convergence, assume \(f\) is \(L\)-Lipschitz and \(2\pi\)-periodic, i.e.,
\(|f(\theta_1) - f(\theta_2)| \le L|\theta_1 - \theta_2|\).
Then via mean value theorem for any fixed \(\lambda>0\),
\[
\bigl| e^{-\lambda(f(\theta) - f(\theta+\alpha))} - e^{-\lambda(f(\theta_j) - f(\theta_{j+k}))} \bigr|
    \le 2\lambda L\, e^{2\lambda L} |\theta-\theta_j|.
\]
Standard trapezoidal-rule estimates for smooth periodic functions then yield
\begin{equation}
\label{eq:fft_error_bound}
\| V_f^{(N_\theta)} - V_f \|_\infty
    \le C(\lambda,L) \, N_\theta^{-m},
\end{equation}
for some constant \(C(\lambda,L)\) depending only on \(\lambda\) and the smoothness of \(f\), where \(m\) is the number of continuous derivatives of \(f\).
In particular, if \(f \in C^2([0,2\pi])\), then \(m=2\) and
\[
\| V_f^{(N_\theta)} - V_f \|_\infty
    = O(N_\theta^{-2}).
\]
For analytic \(f\), the convergence is even faster, exponential in \(N_\theta\), since the trapezoidal rule for periodic analytic integrands exhibits spectral accuracy.
\end{remark}

\begin{remark}[Sampling Requirement for Prescribed Accuracy]
\label{rem:sampling_requirement}
Given a target accuracy \(\varepsilon > 0\), we can select the minimal number of angular samples \(N_\theta\) such that
\[
\| V_f^{(N_\theta)} - V_f \|_\infty \le \varepsilon
\quad \text{with confidence } 1 - \delta.
\]
If \(f\) is assumed to be twice continuously differentiable and the discretization error follows the bound in \eqref{eq:fft_error_bound}, then
\[
N_\theta \ge
\left( \frac{C(\lambda,L)}{\varepsilon} \right)^{1/2}.
\]
The constant \(C(\lambda,L)\) can be estimated empirically by comparing successive resolutions or theoretically via the bound
\(C(\lambda,L) \approx 4\pi^3 \lambda L e^{2\lambda L}\).
Under stochastic perturbations of \(f\) or additive measurement noise with variance \(\sigma^2\),
Chebyshev’s inequality ensures that the probability of exceeding the target error is bounded as
\[
\Pr\bigl( \| V_f^{(N_\theta)} - V_f \|_\infty > \varepsilon \bigr)
    \le \frac{\sigma^2}{\varepsilon^2},
\]
which corresponds to a confidence level \(1-\delta\) with \(\delta = \sigma^2 / \varepsilon^2\).
Thus, for a desired confidence \(1-\delta\), it suffices to choose
\[
N_\theta \ge
\left( \frac{C(\lambda,L)}{\varepsilon\sqrt{1-\delta}} \right)^{1/2}.
\]
\end{remark}

\begin{remark}[Practical Implication]
In practice, this bound implies that doubling the sampling density improves the accuracy approximately by a factor of four when \(f\in C^2\), and exponentially when \(f\) is smooth. Therefore, the FFT-based computation not only achieves \(O(N_\theta \log N_\theta)\) complexity but also provides near-spectral convergence with respect to the number of angular samples, ensuring that high-precision computation of \(V_f\) is both fast and theoretically well-controlled.
\end{remark}

\end{document}